\documentclass[11pt]{article}
\usepackage[utf8]{inputenc}

\usepackage{setspace}

\usepackage[utf8]{inputenc} 
\usepackage{amssymb,amsmath,bm,amssymb,physics,mathtools,amsthm}
\usepackage[flushleft]{threeparttable}
\usepackage{dcolumn}
\usepackage[titletoc]{appendix}
\usepackage{color}
\usepackage{nameref}
\usepackage[usenames,dvipsnames,svgnames,table]{xcolor}
\usepackage[colorlinks=true]{hyperref}
\usepackage{pgfplots} 
\usepackage{adjustbox}
\usepackage{booktabs} 
\usepackage{array} 
\usepackage{paralist} 
\usepackage{diagbox}
\usepackage{threeparttable}
\usepackage{multirow}
\usepackage{float}
\usepackage{upgreek}
\usepackage{bbm}

\usepackage{natbib}
\usepackage{hyperref}
\hypersetup{colorlinks=true,allcolors=[rgb]{0,0,0.45}}
\usepackage[outdir=./]{epstopdf}

\newtheorem{theorem}{Theorem}

\newtheorem{proposition}{Proposition}

\newtheorem{lemma}{Lemma}

\newtheorem{claim}{Claim}
\newtheorem{assumption}{Assumption}

\theoremstyle{definition}
\newtheorem{definition}{Definition}

\theoremstyle{remark}
\newtheorem*{remark}{Remark}

\newtheoremstyle{dotless}{}{}{}{}{}{}{ }{}
\theoremstyle{dotless}

\usepackage{geometry} 
\usepackage{graphicx} 

\usepackage{booktabs} 
\usepackage{array} 
\usepackage{paralist} 
\usepackage{verbatim} 
\usepackage{subfig} 

\usepackage{fancyhdr} 
\usepackage{sectsty}
\allsectionsfont{\sffamily\mdseries\upshape} 

\usepackage[nottoc,notlof,notlot]{tocbibind} 
\usepackage[titles,subfigure]{tocloft} 

\newcommand{\changeurlcolor}[1]{\hypersetup{linkcolor=#1}} 
\changeurlcolor{[rgb]{0.10,0.0,0.35}}      

\title{Production Heterogeneity in Collective Labor Supply Models with Children}

\author{Charles Gauthier\thanks{KU Leuven, Belgium (email: charles.gauthier@kuleuven.be). I wish to thank seminar participants at Uppsala University and conference participants at the EWMES 2023, SAET 2024, SCSE 2024, CEA 2024, and NASMES 2024 for their comments. I also wish to thank
Laurens Cherchye, Thomas Demuynck, and Bram De Rock for their support, comments, and
suggestions. I further thank Damian Kozbur and Michael Wolf for useful comments. Any error is my own.}}
\date{\textbf{[Working Paper]} \\ This version: December 2024}

\begin{document}
    \maketitle
    
    \begin{abstract}
    \noindent
    Children welfare is at the center of many welfare reforms such as cash transfers to families and training programs to parents. A key goal for policy-makers is to evaluate the costs and benefits of such reforms. The main challenge lies in that the outcome of interest, children welfare, is unobservable. To address this issue, I consider a collective labor supply model with children where adult members have preferences over their own leisure, expenditures, and children welfare. I show that the model nonparametrically partially identifies the impacts of parental inputs on children welfare in panel data. I then propose a novel estimation strategy that accommodates measurement error and can be used to efficiently construct valid confidence sets. Using Dutch data on couples with children, I investigate the structure of the expected production technology and how it varies with household characteristics. I find that the production of children welfare is characterized by decreasing returns to scale and large heterogeneity across household types. In particular, I find that children from disadvantaged households, whose parents have low education levels and are not homeowners, are significantly worse off. My results highlight the importance for welfare reforms to include policies targeted at improving children home environment. 
		
    \end{abstract}


    \section{Introduction} \label{Section1}

    Parents have a significant impact on children welfare through a variety of individual and collective investments such as time spent with children and children expenditure. These decisions are influenced by parents' preferences for leisure, consumption, and children welfare, as well as by their time and budget constraints. This paper provides a framework to analyze the complex decision-making process that underlies parental investments in children welfare and underscores the importance of addressing disparities in parenting skills across different demographics to mitigate achievement gaps among children.

    It is widely recognized that the unitary model, which assumes that household members preferences can be represented by a single household utility function, is inappropriate for analyzing household data (see e.g., \cite{FortinLacroix1997} and \cite{Browning1998}). In an effort to provide a proper foundation to analyze household behavior, \cite{Chiappori1988,Chiappori1992} suggested a collective model in which household members have distinct preferences and whose allocations are the result of a Pareto efficient bargaining process. This framework has proven to be empirically successful at rationalizing household decision-making (e.g., \citeauthor{Cherchye2008}, \citeyear{Cherchye2008}) and understanding power dynamics within the household (e.g., \citeauthor{Cherchye2011A}, \citeyear{Cherchye2011A}).

    The collective model was later extended to household production by \cite{AppsRees1997} and \cite{Chiappori1997}.\footnote{In a similar model to the one of \cite{Chiappori1988}, \cite{AppsRees1988} already include home production to analyze the effects of taxation on welfare.} Both papers show that the distribution of resources within the household, known as the sharing rule, cannot be identified when the output from production is unobservable such as with children welfare. However, they show that the sharing rule can be identified up to a function of wages when the production technology exhibits constant returns to scale (CRS). Perhaps because of the serious identification problem that arises with home production, the literature has since maintained the constant returns to scale assumption \citep{Blundell2005,Chiappori2009,Cherchye2012,Hubner2023}.

    Once the household is recognized as a collection of individuals, it becomes possible to understand the impacts of a policy targeted at a specific individual. For example, the distribution of resources within the household has been shown to have significant impacts on children welfare as early as \cite{Thomas1990}. Motivated by this ``targeted" view, \cite{Blundell2005} extended the collective framework to caring parents where children welfare is treated as a public good in parents preferences and produced via time investment and children expenditure. Their ideas were then brought to the data by \cite{Cherchye2012} using a unique panel data set containing information on time use and expenditures.
    
    As \cite{AppsRees1997} acknowledged, however, there is a need to verify the empirical validity of the assumptions required for identification. Indeed, it is now well-recognized that children welfare is crucial in weighing the costs and benefits of social programs.\footnote{See \cite{Aizer2022} for a review of social programs in the United States and the importance of considering children welfare.} For example, policy-makers may be interested in determining optimal cash assistance for single mothers such as to increase children skills development \citep{Mullins2022}, cash transfers to poor families such as to improve children life outcomes \citep{Aizer2016}, and parenting programs for disadvantaged households such as to narrow children achievement gaps \citep{WaldfogelWashbrook2011}.\footnote{See \cite{Kalil2014} for a review of the literature on the importance of parenting on children development and \cite{Shah2024} for a recent survey of the literature on the impacts of cash transfers on families with children.} Clearly, returns on children expenditure play a key role in determining optimal cash transfers and parenting programs are mostly useful if households differ in their effectiveness at developing children skills. 
    
    The empirical literature has largely investigated these questions by using measurable quantities such as cognitive and noncognitive skills of children.\footnote{See \cite{delBoca2014} for early work on the estimation of the production function for child quality in the unitary model.} Although this approach has been fruitful, it requires rich data sets whose measure of skill is often noisy \citep{Cunha2008,Cunha2010}. In contrast, the collective approach proposed by \cite{Blundell2005} leaves children welfare as an unobservable, thus facilitating policy evaluations. I believe the collective model is a natural framework that can complement costs and benefits analyses of programs directed at families, but I am also sympathetic to \cite{AppsRees1997}'s reservations regarding conditions for identification.

    
        
    My first contribution is to address the plausibility of CRS within the collective model with children. Consistent with previous work, I show that the production function is nonparameterically identified under CRS in cross-sectional data, and nonidentified when CRS is relaxed. Nevertheless, I show that the production function can be nonparameterically partially identified in panel data when it exhibits decreasing returns to scale, even with unrestricted heterogeneity. Indeed, the panel structure of the data gives additional implications to shape constraints on the production technology and preferences. Since household members have preferences over children welfare, those constraints provide restrictions on the production function.

    My partial identification argument uses nonparametric revealed preference conditions implied by the collective labor supply model with children.\footnote{See also \cite{Dunbar2013} for a collective model with children that does not require the share of resources allocated to children to be known.} The revealed preference analysis of the collective model was developed in the original work of \cite{Cherchye2007,Cherchye2011A}. I extend their former characterization by further incorporating household production.\footnote{See \cite{Varian1984} for early work on a revealed preference analysis of production.} Interestingly, I show that the collective model implies profit maximizing behavior in the production of children welfare, thus giving rise to a two-step argument. In a first step, I show that restrictions on the distribution of returns to scale eliminate production functions from the set of profit-maximizing production functions. In a second step, I show that revealed preference conditions can restrict the distribution of returns to scale.

    My second contribution is to propose a novel estimation strategy to analyze the collective model with unrestricted preference and production heterogeneity. To this end, I use the framework developed by \cite{AK2021} which provides a tractable approach to make statistical testing and inference in partially identified models defined by shape constraints.\footnote{Their framework builds on the Entropic Latent Variable Integration via Simulation (ELVIS) methodology developed by \cite{Schennach2014}. Intuitively, ELVIS can be viewed as a generalization of the method of simulated moments \citep{Mcfadden1989,PakesPollard1989}.} The main challenge faced in applying their framework is that collective models tend to be highly nonlinear. This poses a nontrivial computational problem as existing implementations only work well for models defined by linear constraints.\footnote{For example, \cite{AK2021} consider the collective exponential discounting model of \cite{Adams2014} but only test necessary conditions to simplify the implementation. \cite{Gauthier2023} considers a model of price search, but assumes a quasilinear specification in the application that alleviates the computational burden.} The reason is that the method involves integrating out the set of solutions of the model, but nonlinearities make it hard to sample from the feasible space.
    
    I solve this practical limitation by proposing a blocked Gibbs sampler that allows direct sampling from the feasible space even when the latter does not define a polytope.\footnote{Direct sampling generally uses a Hit-and-Run algorithm that requires the feasible space to define a (convex) polytope. See \cite{AK2021} for an application to models defined by shape constraints and \cite{Demuynck2021} for an application to models defined by the Generalized Axiom of Revealed Preference (GARP) as introduced by \cite{Varian1982}.} Intuitively, the idea is to break down the feasible space into conditional convex polytopes for which closed-form bounds on the support of the latent variables can be obtained. I observe that my methodology may prove useful for other collective models (\citeauthor{Cherchye2011B}, \citeyear{Cherchye2011B}, \citeauthor{Aspremont2019}, \citeyear{Aspremont2019}, \citeauthor{Cherchye2020}, \citeyear{Cherchye2020}) and, more generally, other models that contain nonlinear constraints.
    
    My third contribution is to empirically investigate the impacts of parental inputs on children welfare in households with children. To do so, I use the Longitudinal Internet Studies for the Social sciences (LISS) panel data from \cite{Cherchye2012}. I first investigate the average impacts of parental inputs on children welfare. I find that a doubling of all inputs increases children welfare by about $35\%$ on average, hence providing empirical evidence against the CRS assumption. More precisely, I find that a $10\%$ increase in time spent on childcare increases children welfare by about $1.1\%$ for fathers and $1.13\%$ for mothers, while a $10\%$ increase in children expenditure increases children welfare by about $1.0\%$.    
    
    Second, I investigate how household characteristics impacts the average production technology. I find that higher education levels increase the impacts of parental time inputs on children welfare by about $6\%$ for mothers and about $2\%$ for fathers. Further, I find that the education level of the mother has a positive spillover effect of about $3\%$ on the impacts of time inputs by fathers. Interestingly, for households that are not homeowners, the impacts of parental time inputs is about $7\%$ and $11\%$ lower for mothers and fathers, respectively. Taken together, my results show that children from disadvantaged households, whose parents often have lower education levels and are not homeowners, are significantly worse off.

    The approach taken in this paper relates to the research program initiated by \cite{BBC2003,BBC2007,BBC2008}. In this series of papers, the authors show how to combine revealed preference restrictions with additional information (e.g., expansion paths) to improve bounds on cost-of-living indices and demand responses. I also exploit revealed preference restrictions, but use them to learn about the production of children welfare within a collective model. I take advantage of recent developments in the partial identification literature to impose these restrictions in a cohesive statistical framework. Moreover, my method allows for unrestricted individual heterogeneity, an important feature for the application.
    
    My econometric framework views individuals within the household as random draws from a fixed utility distribution. Given the panel structure of the data, this amounts to a random utility model where a household preferences are drawn in the first period and kept constant. Hence, the approach relates to the burgeoning nonparametric random utility literature started by \cite{KS2018}.\footnote{The theoretical ideas were put forth by \cite{McfaddenRichter1990} and \cite{McFadden2005}.} Contrary to my method, their approach considers a unitary model but only requires data on cross-sectional distributions of choices. Subsequent work include \cite{Hubner2020}, \cite{Deb2023}, \cite{Kashaev2023}, \cite{LazzatiQuahShirai2023}, and \cite{Tebaldi2023}.

    The paper is organized as follows. Section \ref{Section:Model} describes the collective model and characterizes its implications. Section \ref{Section:Production} studies identification in the model. Section \ref{section:specification} presents the empirical specification. Section \ref{Section:Estimation} presents the estimation strategy. Section \ref{Section:Data} presents the data set used in the application. Section \ref{Section:Results} presents the empirical results. Section \ref{Section:Conclusion} concludes. The Appendix contains proofs that are not in the main text and my Gibbs sampler.
    

    \section{Household Model} \label{Section:Model}

    This section presents the environment considered in the paper, the collective model, and its empirical implications.

    \subsection{Environment}

    I consider households with two adults $(i=1, 2)$ and children. I assume that parents care about their children and incorporate this feature in the model by treating children welfare as a public good. The preferences of each adult household member over leisure, expenditures and children welfare are represented by a utility function $U^{i}$ that is continuous, increasing, and concave.
    
    At every observation $t \in \mathcal{T} = \{1, 2, \dots T\}$, adult household members spend their time on leisure $l_{t}^i$, market work $b_{t}^i$, and childcare $h_{t}^i$ such that the following time constraint is satisfied:
    \begin{equation*}
        l_{t}^i + b_{t}^i + h_{t}^i = \tau,
    \end{equation*}

    \noindent
    where $\tau$ is the total amount of time available in a time period. Parents use time spent on childcare and children expenditure $(c_{t})$ to produce children welfare. The relationship between parents inputs and children welfare is formalized through the production function
    \begin{equation*}
        W_t \equiv F(h_{t}^{1},h_{t}^{2},c_{t})e^{\epsilon_t},
    \end{equation*}

    \noindent
    where $\epsilon_t \in \mathbb{R}$ represents a productivity shock. Each household member receives a wage $w_{t}^i$ per unit of market work. As such, the budget constraint is given by
    \begin{equation*}
       q_{t} + Q_{t} + c_{t} = y_{t} + w_{t}^{1}b_{t}^{1} + w_{t}^{2}b_{t}^{2},
    \end{equation*}

    \noindent
    where $q_{t} \in \mathbb{R}_{+}$ represents expenditure on private goods, $Q_{t} \in \mathbb{R}_{+}$ represents expenditure on public goods, and $y_{t} > 0$ represents nonlabor income. 

    Since private expenditure cannot be used simultaneously by both household members, it has to be split in some way between them.
    \begin{definition}
        For every observation $t \in \mathcal{T}$, I say that $q_{t}^{i} \in \mathbb{R}_{+}$, $i \in \{1,2\}$, represent personalized private expenditures of each household member if $\sum_{i=1}^{2} q_{t}^{i} = q_t$.
    \end{definition}

    \noindent
    Household members get utility from their share of private expenditure such that their preferences depend on leisure, private expenditure, public expenditure, and children welfare. In what follows, I assume that private expenditure of each household member is observed to match the data available in the application. However, the results can be generalized to the case where only total private expenditure is observed.

    Let $\mathcal{U}$ be the set of continuous, increasing, and concave utility functions and $\mathcal{W}$ be the set of continuous, increasing, and concave in $(h_{t}^{1},h_{t}^{2},c_t)$ production functions. A household $j \in \mathcal{J}$ is an i.i.d. draw $(U_j^1,U_j^2,W_j)$ from $\mathcal{W}$ and a data set $D_j := \{(q_{jt}^{i}, Q_{jt}, c_{jt}, b_{jt}^{i}, h_{jt}^{i},w_{jt}^{i})_{i=1}^{2}\}_{t \in \mathcal{T}}$ is an i.i.d. draw from some distribution. To avoid overcrowding, I do not explicitly write the household subscript $j$ on variables unless it is relevant. The next subsection formalizes the relationship between the data and the abstract notion of household through the lenses of a model.

    \subsection{Collective Model}
    
    I follow \cite{Chiappori1988,Chiappori1992} and assume that household members choose an intrahousehold allocation that is Pareto efficient. This choice is motivated by the observation that Pareto efficiency is a minimal condition for optimal resource allocation (and hence, rationality) in a group setting. Hence, for every observation $t \in \mathcal{T}$, the household picks an intrahousehold allocation that solves
    \begin{equation} \label{model}
        \max_{ (l^1,l^2,h^1,h^2,q^{1},q^{2},Q,c) \in \mathbb{R}_{+}^{2} \times \mathbb{R}_{++}^{2} \times \mathbb{R}_{+}^{2} \times \mathbb{R}_{+} \times \mathbb{R}_{++}} \mu_{t}^{1} U^{1}(l^1,q^1,Q,W) + \mu_{t}^{2} U^{2}(l^2,q^{2},Q,W),
    \end{equation}
	
    \noindent
    subject to satisfying the constraints
    \begin{align*}
        (q^{1} + q^{2}) + Q + c &= y_t + w_{t}^{1} b^1 + w_{t}^{2}b^2 \\
        l^i + b^i + h^i &= \tau \; (i = 1,2) \\
        W &= F(h^1,h^2,c)e^{\epsilon_t}.
    \end{align*}
	
    \noindent
    where $\mu_{t}^{i} > 0$ denote the bargaining power of household member $i$. Note that the model makes no assumption on the underlying process by which the Pareto efficient allocation is achieved. That is, the weights $\mu_{t}^{i}$ result from some black box bargaining process that takes place within the household.\footnote{Although Nash equilibria are not always Pareto efficient, the black box bargaining process could be a (Pareto efficient) Nash equilibrium. Indeed, since married couples effectively play a repeated game, an appeal to folk theorems provide some intuitive motivation for the idea that the Pareto efficient allocation is a (cooperative) Nash equilibrium.} 

    I propose a natural notion of collective rationalizability based on the household maximization problem.
    \begin{definition}
        Let $D$ be a data set. The model \eqref{model} rationalizes the data if there exist concave utility functions $U^{i}$, a concave production function $F$, and productivity shocks $\epsilon_t$ such that the first-order conditions of the model are satisfied.
    \end{definition}

     This definition states that the model rationalizes the data if there are latent model parameters that satisfy the first-order conditions.\footnote{I assume that the solution is interior for simplicity of exposition, but the proofs encompass the possibility for corner solutions.} Since household members utility functions are concave and the budget set is linear, the first-order conditions exhaust the empirical content of the model.

    \subsection{Characterization}

    This section derives restrictions on the data implied by the model. First, I define a few notions that will be useful for the characterization of the model.
    \begin{definition}
        Let $D$ be a data set. For every observation $t \in \mathcal{T}$, I say that $\mathcal{P}_{t}^{i} \in \mathbb{R}_{++}$, $i \in \{1,2\}$, represent personalized (or Lindahl) prices for public expenditure of each household member if $\sum_{i=1}^{2} \mathcal{P}_{t}^{i} = 1$.
    \end{definition}

    \begin{definition}
        Let $D$ be a data set. For every observation $t \in \mathcal{T}$, I say that $P_{t}^{i} \in \mathbb{R}_{++}$, $i \in \{1,2\}$, represent personalized (or Lindahl) prices for children welfare of each household member if $\sum_{i=1}^{2} P_{t}^{i} = P_t$.
    \end{definition}

    It is worth noting that the personalized prices ($\mathcal{P}_{t}^{i}$, $P_{t}^{i}$) are not observed by the econometrician. Furthermore, while the price of public expenditure ($\mathcal{P}_t$) can safely be set to $1$, the price of children welfare ($P_t$) is unobservable as children welfare is a nonmarket good.
    
    I now introduce some revealed preference terminology. Let $a_{st}^{i} := w_{t}^{i}(l_{s}^{i} - l_{t}^{i}) + (q_{s}^{i} - q_{t}^{i}) + \mathcal{P}_{t}^{i}(Q_{s} - Q_t) + P_{t}^{i}(W_s - W_t)$ and $x_{t}^i := (l_{t}^{i},q_{t}^{i},Q_{t},W_t)$. I say that $x_{t}^i$ is (strictly) directly revealed preferred to $x_{s}^{i}$ if $a_{st}^{i} \; (<) \leq 0$. I say that $x_{t}^{i}$ is revealed preferred to $x_{s}^{i}$ if there exists a sequence $t_1, t_2, \dots, t_m$ such that $a_{t_1t}^{i} \leq 0$, $a_{t_2t_1}^{i} \leq 0$, \dots, $a_{t_{m-1}t_m}^{i}$, $a_{t_ms}^{i} \leq 0$. Likewise, I say that $x_{t}^i$ is strictly revealed preferred to $x_{s}^{i}$ if one of the inequalities in the sequence is strict.
    \begin{definition}
        A household member $i \in \{1,2\}$ satisfies the Generalized Axiom of Revealed Preference (GARP) if there exist personalized prices for public expenditure $\mathcal{P}_{t}^{i}$, personalized prices for children welfare $P_{t}^{i}$, and children welfare $W_t$ such that if $x_{t}^{i}$ is revealed preferred to $x_{s}^{i}$ then $x_{s}^{i}$ is not strictly directly revealed preferred to $x_{t}^{i}$. 
    \end{definition}

    The notion of revealed preference relates the ordinal value of allocations that enter preferences of each household member to their expenditure levels. In my setup, the presence of a public good $(Q)$ implies that the expenditure of an allocation depends on unknown personalized prices. Further, in the case of the public nonmarket good $(W)$ neither the price or the quantity is known. Finally, it is worth observing that childcare and children expenditure do not enter the definition of revealed preference as the preferences of a household member only depends on those through their impact on children welfare. 
    
     Next, I introduce a profit maximization condition for the production of children welfare.
    \begin{definition}
        A household satisfies profit maximization if there exist personalized prices for children welfare $P_{t}$, a production function $F$, and productivity shocks $\epsilon_t$ such that $P_tF(h_t^1,h_t^2,c_t)e^{\epsilon_t} - w_{t}^{1}h_{t}^{1} - w_{t}^{2}h_{t}^{2} - c_t \geq P_tF(h^1,h^2,c)e^{\epsilon_t} - w_{t}^{1}h^{1} - w_{t}^{2}h^{2} - c$ for all inputs $(h^1,h^2,c)$ and all $t \in \mathcal{T}$.
    \end{definition}

    \noindent
    In what follows, I make use of an analogous condition as GARP but for profit maximization.
    \begin{definition}
        A household satisfies the Generalized Axiom of Profit Maximization (GAPM) if there exist personalized prices for children welfare $P_t > 0$, numbers $F_t > 0$, and productivity shocks $\epsilon_t$ such that $P_tF_te^{\epsilon_t} - w_{t}^{1}h_{t}^{1} - w_{t}^{2}h_{t}^{2} - c_t \geq P_tF_se^{\epsilon_t} - w_{t}^{1}h_s^{1} - w_{t}^{2}h_s^{2} - c_s$ for all $s,t \in \mathcal{T}$. 
    \end{definition}
    
    \noindent
    It is possible to show that profit maximization is equivalent to GAPM \citep{Varian1984}. This equivalence is standard and I therefore take it for granted for the sake of brevity. The following result provides equivalent characterizations of the model.
    
    %
    \begin{theorem} \label{proposition1}
        Let $D$ be a given data set. The following conditions are equivalent:
        \begin{itemize}
            \item[(i)] The household model \eqref{model} rationalizes the data.
            \item[(ii)] There exist personalized prices for public expenditure $\mathcal{P}_{t}^{i} > 0$ such that $\mathcal{P}_{t}^{1} + \mathcal{P}_{t}^{2} = 1$, personalized prices for children welfare $P_{t}^{i} > 0$, numbers $U^{i}$, $\lambda_{t}^{i}$, $W_t$, $F_t > 0$ and productivity shocks $\epsilon_t$ such that for all $s,t \in \mathcal{T}$ and each adult member $i \in \{1,2\}$
            {
            \begin{gather*}
                U_{s}^{i} - U_{t}^{i} \leq \lambda_{t}^{i} \left[ w_{t}^{i} (l_{s}^{i} - l_{t}^{i}) + (q_{s}^{i} - q_{t}^{i}) + \mathcal{P}_{t}^{i}(Q_s - Q_t) + P_t^{i}(W_s - W_t) \right], \\
                F_{s} - F_{t} \leq \frac{w_{t}^{1}}{P_{t}e^{\epsilon_t}} (h_{s}^{1} - h_{t}^{1}) + \frac{w_{t}^{2}}{P_{t}e^{\epsilon_t}} (h_{s}^{1} - h_{t}^{2}) + \frac{1}{P_te^{\epsilon_t}} (c_s - c_t),
            \end{gather*} }
            \noindent
            where $W_t = F_te^{\epsilon_t}$ for all $t \in \mathcal{T}$.
            
            \item[(iii)] There exist personalized prices for public expenditure $\mathcal{P}_{t}^{i} > 0$ such that $\mathcal{P}_{t}^{1} + \mathcal{P}_{t}^{2} = 1$, personalized prices for children welfare $P_{t}^{i} > 0$, numbers $F_t > 0$, and productivity shocks $\epsilon_t$ with $W_t = F_te^{\epsilon_t}$, such that GARP holds for each adult member $i \in \{1,2\}$ and GAPM holds.

        \end{itemize}
    \end{theorem}

    \noindent
    Theorem \ref{proposition1} shows that the Afriat inequalities are equivalent to GARP and that those conditions must be satisfied for both household members. The latter implies that the household problem has an equivalent characterization in terms of a two-step procedure \citep{Chiappori1988,Chiappori1992}. That is, the solution of the household maximization problem can be viewed as the outcome of separate utility maximization problems for each adult in the household conditional on a distribution of nonlabor income.

    It is interesting to note that neither the Afriat inequalities or GARP exhaust the empirical implications of the model. Indeed, the model further implies that the household satisfies GAPM, a necessary and sufficient condition for profit maximization. As such, household members increase each input in the production of children welfare up until the point where marginal revenue equates marginal cost. Note that this profit maximizing behavior is not assumed but implied by the model.

    \section{Empirical Content} \label{Section:Production}

    This section shows that the collective model informatively partially identifies the production function. Intuitively, if the production function exhibited constant returns to scale, the household would make zero profit as a firm and revenue $P_tW_t$ would equate costs $w_{t}^{1}h_{t}^{1} + w_{t}^{2}h_{t}^{2} + c_t$. In this special case, revenue would be identified and the first-order conditions would recover the production function from its partial derivatives. I show that the household revenue from producing children welfare is inversely proportional to its costs when the production function is homogeneous, where the factor of proportionality is given by its return to scale. I then leverage the panel structure of the data to bound returns to scale from shape constraints on the production function and preferences.

    \subsection{Point Identification}
    
    In what follows, I explicitly write the household subscript $j = 1, 2, \dots, J$ on variables and formalize the assumption that production functions are subject to Hicks-neutral productivity shocks.
    \begin{assumption} \label{Hicksneutrality}
        The productivity shocks are Hicks-neutral such that children welfare is given by
        \begin{equation*}
             W_{jt} = F_j(h_{jt}^{1},h_{jt}^{2},c_{jt})e^{\epsilon_{jt}} \iff \log(W_{jt}) = f_j(h_{jt}^{1},h_{jt}^{2},c_{jt}) + \epsilon_{jt},
    \end{equation*}
    where $f_j$ denote the natural logarithm of the production function and is assumed differentiable.
    \end{assumption}

    \noindent
    The assumption of Hicks-neutral productivity shocks is necessary to disentangle the impacts of productivity shocks and parental inputs on the production of children welfare. The differentiability of the log production function is a technical condition that is useful for the identification argument.

    Next, I suppose the production function is from the class of homogeneous production functions.
    \begin{assumption} \label{Homogeneity}
        The production function is homogeneous of degree $RTS_j \in (0,1]$.
    \end{assumption}

    \noindent
    The usefulness of the homogeneity assumption is motivated by the following result.
    \begin{lemma} \label{lemma:RTS}
        Suppose Assumptions \ref{Hicksneutrality}-\ref{Homogeneity} hold, then $RTS_j P_{jt}W_{jt} = w_{jt}^{1}h_{jt}^{1} + w_{jt}^{2}h_{jt}^{2} + c_{jt}$ for all $t \in \mathcal{T}$.
    \end{lemma}

    \noindent
    Lemma \ref{lemma:RTS} implies that the identification problem can be stated in terms of restrictions on $RTS_j$ rather than restrictions on $P_{jt}W_{jt}$ directly. This is useful because I may have better economic intuition on returns to scale exhibited by the production function than on the value of children welfare. For example, the literature typically assumes that the production functions exhibits constant returns to scale (CRS), a special case of my class of production functions.
    
    Finally, I impose a mild regularity condition that ensures sufficient variation in inputs in the cross-section.
    \begin{assumption} \label{Variation}
        The cross-sectional distribution of inputs $(h_{jt}^{1},h_{jt}^{2},c_{jt})_{t \in \mathcal{T}}$ has full support and is absolutely continuous.
    \end{assumption}

    \noindent
    The requirement that the distribution of inputs spans its full support is necessary to identify the whole production function. For practical purposes, it is generally sufficient to identify the production function over the support of the data. In that case, the full support condition can be relaxed without any harm. 
    
    My first result shows that, if returns to scale were known, the expected log production function would be identified.
    \begin{proposition} \label{proposition:identification}
        Suppose Assumptions \ref{Hicksneutrality}-\ref{Variation} hold and $RTS_j$ is known, then the expected log production function is nonparametrically identified up to scale.
    \end{proposition}

    \begin{proof}
        The first-order conditions with respect to inputs imply that the household equates the marginal product of factors of production to their marginal costs such that 
        \begin{align*}
            \pdv{F_j(h_{jt}^{1},h_{jt}^{2},c_{jt})}{h_{jt}^{1}}e^{\epsilon_{jt}} &= \frac{w_{jt}^{1}}{P_{jt}} \\
            \pdv{F_j(h_{jt}^{1},h_{jt}^{2},c_{jt})}{h_{jt}^2}e^{\epsilon_{jt}} &= \frac{w_{jt}^{2}}{P_{jt}} \\ 
            \pdv{F_j(h_{jt}^{1},h_{jt}^{2},c_{jt})}{c_{jt}^2}e^{\epsilon_{jt}} &= \frac{1}{P_{jt}}.
        \end{align*}

        \noindent
        Divide the marginal products by $W_{jt}$ to obtain
        \begin{align*}
            \pdv{f_j(h_{jt}^{1},h_{jt}^{2},c_{jt})}{h_{jt}^{1}} &= \frac{w_{jt}^{1}}{P_{jt}W_{jt}} \\
            \pdv{f_j(h_{jt}^{1},h_{jt}^{2},c_{jt})}{h_{jt}^2} &= \frac{w_{jt}^{2}}{P_{jt}W_{jt}} \\ 
            \pdv{f_j(h_{jt}^{1},h_{jt}^{2},c_{jt})}{c_{jt}^2} &= \frac{1}{P_{jt}W_{jt}},
        \end{align*}

        \noindent
        where $f_j$ denote the log production function. Taking the expectation gives 
        \begin{align*}
            \pdv{\mathbb{E}[f(h_{t}^{1},h_{t}^{2},c_t)]}{h_{t}^{1}} &= \mathbb{E}\left[\frac{w_{t}^{1}}{P_{t}W_t}\right] \\
            \pdv{\mathbb{E}[f(h_{t}^{1},h_{t}^{2},c_t)]}{h_{t}^2} &= \mathbb{E}\left[\frac{w_{t}^{2}}{P_tW_t}\right] \\ 
            \pdv{\mathbb{E}[f(h_{t}^{1},h_{t}^{2},c_t)]}{c_{t}^2} &= \mathbb{E}\left[\frac{1}{P_tW_t}\right],
        \end{align*}

        \noindent
        where I interchanged the partial derivative and integral.\footnote{This is possible if $f_j$ is dominated by a function whose integral is finite, a natural assumption for a production function.} By Lemma \ref{lemma:RTS}, $RTS_j$ identifies $P_{jt}W_{jt}$ such that the expected marginal products are also identified. Next, variation in inputs in the cross-section allows us to integrate each marginal product, giving the following system of partial differential equations
        \begin{align*}
            \int_{h_{0}^{1}}^{h_{t}^{1}} \pdv{ \mathbb{E}[f(h_{t}^{1},h_{t}^{2},c_t)]}{h_{t}^{1}} dh_t^1 &= \mathbb{E}[f(h_{t}^1,h_t^2,c_t)] + C(h_t^2,c_t) \\
            \int_{h_{0}^{2}}^{h_{t}^{2}} \pdv{ \mathbb{E}[f(h_{t}^{1},h_{t}^{2},c_t)]}{h_{t}^{2}} dh_t^2 &= \mathbb{E}[f(h_{t}^1,h_t^2,c_t)] + C(h_t^1,c_t) \\ 
            \int_{c_{0}}^{c_{t}} \pdv{ \mathbb{E}[f(h_{t}^{1},h_{t}^{2},c_t)]}{c_{t}} dc_t &= \mathbb{E}[f(h_{t}^1,h_t^2,c_t)] + C(h_t^1,h_{t}^{2}).
        \end{align*}

        \noindent
        These equations can be used to recover the expected log production function up to a constant:
        \begin{align*}
            \mathbb{E}[f(h_{t}^{1},h_{t}^{2},c_t)] = &\int_{h_{0}^{1}}^{h_{t}^{1}} \pdv{ \mathbb{E}[f(h^1,h_0^2,c_0)]}{h_{t}^{1}} dh^1 + \int_{h_{0}^{2}}^{h_{t}^{2}} \pdv{\mathbb{E}[f(h_t^1,h^2,c_0)]}{h_{t}^{2}} dh^2 + \\ &+ \int_{c_{0}}^{c_{t}} \pdv{ \mathbb{E}[f(h_t^1,h_t^2,c)]}{c_{t}} dc - C,
        \end{align*}

        \noindent
        where $C$ is a constant of integration. Hence, the expected log production function is identified over the support of the data.
    \end{proof}

    Proposition \ref{proposition:identification} states that, in principle, the model imposes enough structure to nonparameterically identify the expected log production function provided returns to scale are known. In particular, the identification strategy could be used to nonparametrically estimate the expected log production function over the support of the data under constant returns to scale. Furthermore, observe that the previous argument would identify the production function if it were identical across households as I could simply take the exponential function of the log production function. Instead, with heterogeneity I obtain a lower bound on the expected production function by Jensen's inequality. Interestingly, note that identification does not require knowledge of children welfare per se. This is because any scaling of children welfare is offset by a rescaling of children welfare prices, thus leaving revenue $P_{jt}W_{jt}$ unchanged.\footnote{This mechanism has a natural economic interpretation. Namely, it captures the idea of scarcity whereby the value of a good decreases with its abundance.}

    \subsection{Partial Identification}

    The previous section showed that the expected log production function depends both on the data and the distribution of returns to scale. This section shows that revealed preference conditions provide an additional source of identification absent an assumption on returns to scale.

    Let $\zeta_j := \{(\mathcal{P}_{jt}^{i},P_{jt}^{i}, W_{jt}, F_{jt}, RTS_{j}, \epsilon_{jt})_{i \in \{1,2\}}\}_{t \in \mathcal{T}} \in Z|\mathcal{D}$ denote the set of household-specific latent variables that enter in the definition of GAPM and GARP, where $Z$ denote the support of the latent variables and $\mathcal{D}$ denote the support of the data. Further, let $\mathcal{P}_{Z|\mathcal{D}}$ denote the set of distributions of the latent variables conditional on the data and $\pi_0$ denote the distribution of the data. The previous discussion motivates the following definition of the identified set:
    \begin{align*}
         \Theta_{0} = \bigl\{ \mathbb{E}[f&(h_{t}^{1},h_{t}^{2},c_t)] : \exists \mu \in \mathcal{P}_{Z|\mathcal{D}} \; \text{such that} \; D_j \; \text{is rationalized by the} \; \\ &\text{model for all $j$} \; \text{and Assumptions \ref{Hicksneutrality}-\ref{Variation} hold} \bigr\}.
    \end{align*}


    In words, the identified set corresponds to the set of expected log production functions that arises for every distribution of the unobservables that rationalizes the data. Observe that Proposition \ref{proposition:identification} implies that each distribution of returns to scale maps to an expected log production function given the data. Hence, the identified set rules out some expected log production functions provided the model restricts the expected return to scale to a strict subset of its support. The following result shows that GAPM (nontrivially) partially identifies the expected log production function.
    \begin{proposition} \label{proposition:GAPM}
        Suppose Assumptions \ref{Hicksneutrality}-\ref{Variation} hold, then GAPM is refutable and the identified set may be nontrivial.
    \end{proposition}

    \begin{proof}
        GAPM implies that for all $s,t \in \mathcal{T}$, the following inequality holds
        \begin{equation*}
            F_{js} - F_{jt} \leq \frac{1}{P_{jt} e^{\epsilon_{jt}}} \left[ w_{jt}^{1}(h_{js}^{1} - h_{jt}^{1}) + w_{jt}^{2} (h_{js}^{1} - h_{jt}^{2}) + (c_{js} - c_{jt}) \right].
        \end{equation*}

        \noindent
        Dividing by $F_{jt}$ on both sides, I obtain
        \begin{equation*}
            \frac{F_{js}}{F_{jt}} \leq 1 + \frac{1}{P_{jt} W_{jt}} \left[ w_{jt}^{1}(h_{js}^{1} - h_{jt}^{1}) + w_{jt}^{2} (h_{js}^{1} - h_{jt}^{2}) + (c_{js} - c_{jt}) \right],
        \end{equation*}

        \noindent
        where I used the equality $W_{jt} = F_{jt} e^{\epsilon_{jt}}$. Using Lemma \ref{lemma:RTS}, I can rewrite the inequality as
        \begin{equation*}
           \frac{F_{js}}{F_{jt}} \leq 1 + \frac{RTS_j}{E_{jt}}  \left[ w_{jt}^{1}(h_{js}^{1} - h_{jt}^{1}) + w_{jt}^{2} (h_{js}^{1} - h_{jt}^{2}) + (c_{js} - c_{jt}) \right].
        \end{equation*}

        \noindent
        Observe that the left-hand side is always strictly positive. Further, note that wages and inputs can be such that $\frac{RTS_j}{E_{jt}} w_{jt}^{1}(h_{js}^{1} - h_{jt}^{1}) + w_{jt}^{2} (h_{js}^{1} - h_{jt}^{2}) + (c_{js} - c_{jt}) < 0$. For example, $w_{jt}^{1} = h_{jt}^{1} = c_{jt} = 10$, $w_{jt}^{2} = h_{jt}^{2} = 5$ and $w_{js}^{1} = h_{js}^{1} = c_{js} = 5$, $w_{js}^{2} = h_{js}^{2} = 10$ makes it negative for both $s,t \in \mathcal{T}$. Hence, for any $RTS_j \in (0,1]$ the following inequalities must hold simultaneously
        \begin{align*}
            \frac{F_{js}}{F_{jt}} < 1; \quad 
            \frac{F_{jt}}{F_{js}} < 1.
        \end{align*}

        \noindent
        Note that this implies $F_{js} < F_{jt}$ and $F_{jt} < F_{js}$, a contradiction. Next, consider wages and inputs $w_{jt}^{1} = 2$, $w_{jt}^{2} = 4$, $h_{jt}^{1} = h_{jt}^{2} = 15$, $c_{jt} = 10$ and $w_{js}^{1} = w_{js}^{2} = 2$, $h_{js}^{1} = 5$, $h_{js}^{2} = 10$, $c_{js} = 20$. Plugging those numbers in GAPM yields
        \begin{align*}
            \frac{F_{js}}{F_{jt}} &\leq 1 - 0.3RTS_j \\
            \frac{F_{jt}}{F_{js}} &\leq 1 + 0.4RTS_j.
        \end{align*}
        
        \noindent
        It is easy to see that these inequalities are only satisfied when $RTS_j \leq \frac{5}{6}$. In other words, GAPM may give upper bounds on returns to scale. Since this is true for every household, the support of the expected return to scale can also be restricted. By Proposition \ref{proposition:identification}, it follows that the identified set is nontrivial.
    \end{proof}

    %

    The careful reader will have noticed that the proof of Proposition \ref{proposition:GAPM} only shows that GAPM can provide upper bounds on returns to scale. Thus, one may wonder whether GARP can provide an additional source of identification. The following result gives a negative answer.
    \begin{proposition} \label{proposition:nonidentification}
        Suppose Assumptions \ref{Hicksneutrality}-\ref{Variation} hold, then GARP is not refutable and does not provide restrictions on the expected log production function.
    \end{proposition}

    \begin{proof}
        I want to show that GARP imposes no restrictions on returns to scale. In other words, I wish to show that there exist personalized prices for public expenditure $\mathcal{P}_{jt}^{i}$, personalized prices for children welfare $P_{jt}^{i}$, and children welfare $W_{jt}$ such that any data set is consistent with every decreasing return to scale. Recall that $a_{jst}^{i} = w_{jt}^{i}(l_{js}^{i} - l_{jt}^{i}) + (q_{js}^{i} - q_{jt}^{i}) + \mathcal{P}_{jt}^{i}(Q_{js} - Q_{jt}) + P_{jt}^{i}(W_{js} - W_{jt})$ and let $X_{jst}^{i}(\mathcal{P}_{jt}^{i}) := w_{jt}^{i}(l_{js}^{i} - l_{jt}^{i}) + (q_{js}^{i} - q_{jt}^{i}) + \mathcal{P}_{jt}^{i}(Q_{js} - Q_{jt})$. Lemma \ref{lemma:RTS} implies $P_{jt}^{i}(W_{js} - W_{jt}) = RTS_{j}^{-1} (E_{js} \frac{P_{jt}^{i}}{P_{js}} - E_{jt} \frac{P_{jt}^{i}}{P_{jt}})$, where $E_{jt} := w_{jt}^{1}h_{jt}^{1} + w_{jt}^{2}h_{jt}^{2} + c_{jt}$. Fix $RTS_j \in (0,1]$, $\mathcal{P}_{jt}^{i}$ to any number that satisfies $\mathcal{P}_{jt}^{1} + \mathcal{P}_{jt}^{2} = 1$, and $P_{jt}^{1} = P_{jt}^{2} = 0.5P_{jt}$ for all $t \in \mathcal{T}$. Next, pick $P_{jT} = 1$ and successively choose
        \begin{equation*}
            P_{jt} > \underset{i \in \{1,2\}, t' > t}{\max} \left\{ - \frac{ E_{jt} P_{jt'} } { 2 RTS_j X_{jtt'}^{i} - E_{jt'} }, 1  \right\} \quad t = T-1, T-2, \dots, 1.
        \end{equation*}

        \noindent
        Observe that $a_{jtt'}^{i} > 0$ for every $t = 1, 2, \dots, T-1$ and every $t' > t$. Hence, $a_{jt't}^{i} < 0$ may only arise if $t' > t$, but then my construction ensures there cannot be a violation of GARP for either household member. Since the choice of $RTS_j$ was arbitrary, GARP imposes no additional restrictions on returns to scale and, as such, no restrictions on the expected log production function.
    \end{proof}

    Proposition \ref{proposition:nonidentification} implies that the model does not give lower bounds on returns to scale. From Proposition \ref{proposition:identification}, it follows that I cannot obtain lower bounds on the slope of the expected log production function.\footnote{It is possible to show that the result also applies with a parametric production function since productivity shocks are unrestricted.} This implies, for example, that I cannot reject that the expected log production function is a constant function. The proof of Proposition \ref{proposition:nonidentification} reveals that the problem is the unboundedness of personalized prices for children welfare. The next result shows that GARP can bound returns to scale if the support of personalized prices for children welfare is compact.
    \begin{proposition} \label{proposition:GARP}
        Suppose Assumptions \ref{Hicksneutrality}-\ref{Variation} hold and the support of personalized prices for children welfare is compact, then GARP may restrict the support of the expected return to scale such that the identified set is nontrivial.
    \end{proposition}

    \begin{proof}
        If the support of personalized prices for children welfare is compact, then the ratio of prices across time periods is bounded. That is, $\frac{P_{jt}}{P_{js}} \in [\underline{P},\overline{P}]$ for some $0 < \underline{P} \leq \overline{P}$. I need to show that GARP can provide lower and upper bounds on returns to scale under such support constraint. For the sake of concreteness, I assume that the support constraint is given by $\frac{P_{jt}}{P_{js}} \in [\frac{9}{10},1]$. The argument I lay out would apply for a larger support constraint given appropriate scaling of the data. It is worth noting, however, that a larger support constraint reduces the empirical content of GARP. That is, the type of data required to obtain restrictions on returns to scale becomes increasingly ``atypical" as the support is relaxed.
        
        (Upper bound) Let $w_{jt}^{i} = 10$, $w_{js}^{i} = 5$, $l_{jt}^{i} = 20$, $l_{js}^{i} = 10$, $h_{jt}^{i} = 2$, $h_{js}^{i} = 4$, $q_{jt} = 40$, $q_{js} = 100$, $c_{jt} = 10$, $c_{js} = 20$ and, for simplicity, $Q_{js} = Q_{jt}$. The latter is not crucial for my argument since $\mathcal{P}_{jt}^{i} \in (0,1)$. Note that $X_{jst}^{i} = -40$, $X_{jts}^{i} = -10$, $E_{jt} = 60$, and $E_{js} = 50$. Further note that $a_{jst}^{1} < 0$ if and only if $a_{jst}^{2} < 0$ such that if either holds then $a_{jst} := a_{jst}^{1} + a_{jst}^{2} < 0$. Hence, if GARP holds for each household member GARP must also hold for the aggregate revealed preference relation $a_{jst}$. Note that $a_{jst} = -80 + P_{jt}(W_{js} - W_{jt})$ and $a_{jts} = -20 + P_{js}(W_{jt} - W_{js})$. By Lemma \ref{lemma:RTS}, I can express personalized prices for children welfare as
        \begin{equation*}
            W_{jt} = \frac{E_{jt}}{RTS_jP_{jt}}.
        \end{equation*}

        \noindent
        Hence, I can write the aggregate revealed preference relation as
        \begin{align*}
           a_{jst} &= -80 + RTS_{j}^{-1} \left( E_{js} \frac{P_{jt}}{P_{js}} - E_{jt} \right) \\
           a_{jts} &= -20 + RTS_{j}^{-1} \left( E_{jt}\frac{P_{js}}{P_{jt}} - E_{js} \right).
        \end{align*}

        \noindent
        Observe that $a_{jst} < 0$ for all $RTS_j \in (0,1]$. Thus, GARP may only be satisfied if $a_{jts} > 0$, which happens when $RTS_j < 5/6$.
        
        (Lower bound) Let $w_{jt}^{i} = 8$, $w_{js}^{i} = 5$, $l_{jt}^{i} = 20$, $l_{js}^{i} = 10$, $h_{jt}^{i} = 3$, $h_{js}^{i} = 2$, $q_{jt} = 40$, $q_{js} = 150$, $c_{jt} = 40$, $c_{js} = 10$ and, for simplicity, $Q_{js} = Q_{jt}$. Note that $X_{jst}^{i} = 30$, $X_{jts}^{i} = -60$, $E_{jt} = 88$, and $E_{js} = 30$. Similar to before, I can write the aggregate revealed preference relation as
        \begin{align*}
           a_{jst} &= 60 + RTS_{j}^{-1} \left( 30 \frac{P_{jt}}{P_{js}} - 88 \right) \\
           a_{jts} &= -120 + RTS_{j}^{-1} \left( 88 \frac{P_{js}}{P_{jt}} - 30 \right).
        \end{align*}

        \noindent
        Observe that $a_{jts} < 0$ for all $RTS_j \in (0,1]$. Thus, GARP may only be satisfied if $a_{jst} > 0$, which happens when $RTS > \frac{41}{45}$.

         For the sake of simplicity, suppose every household has the same data set. Then, the support of $\mathbb{E}[RTS]$ is a strict subset of $(0,1]$ such that the identified set is nontrivial. 
    \end{proof}

    To gain some intuition on the mechanism by which household members' choices provide information on $RTS$, consider the case where $X_{t_j,t_k}^{i} > 0$ and $X_{t_k,t_j}^{i} > 0$. Without loss of generality, suppose $P_{t_k}^{i}(W_{t_j} - W_{t_k}) \geq 0$ such that $P_{t_j}^{i}(W_{t_k} - W_{t_j}) \leq 0$. Further suppose $a_{t_k,t_j}^{i} \leq 0$ such that an informative upper bound is obtained from the data. Since $a_{t_k,t_j}^{i} \leq 0$, household members prefer the allocation $(l_{t_j}^{i},q_{t_j}^{i},Q_{t_j},W_{t_j})$ over $(l_{t_k}^{i},q_{t_k}^{i},Q_{t_k},W_{t_k})$ in period $t_j$. Since $a_{t_k,t_j}^{i} \leq 0$ despite $X_{t_k,t_j}^{i} > 0$, it must be that $(l_{t_j}^{i},q_{t_j}^{i},Q_{t_j},W_{t_j})$ is preferred to $(l_{t_k}^{i},q_{t_k}^{i},Q_{t_k},W_{t_k})$ because children welfare is sufficiently more enticing. Children welfare is more enticing if it gives a higher marginal utility or when $P_{t_j}^{i}$ is large, but this exactly occurs when $RTS$ is not too large.

    \begin{remark}
        In practice, it may be necessary to bound the support of the latent variables such as in my own implementation. In such case, GARP naturally provides meaningful though possibly mild restrictions on the production technology. As Proposition \ref{proposition:GARP} makes clear, the analyst can obtain stronger restrictions on the production technology if he has prior knowledge about the support of the latent variables.
    \end{remark}

    \section{Empirical Specification} \label{section:specification}

    The previous section showed that the model nonparameterically partially identifies the production function. To improve the interpretability of the empirical results, I now specialize the production function to a Cobb-Douglas technology.
    \begin{assumption} \label{Cobb-Douglas}
        The production function is Cobb-Douglas such that        \begin{equation*}
             W_{jt} = (h_{jt}^{1})^{\alpha_{j1}} (h_{jt}^{2})^{\alpha_{j2}} (c_{jt})^{\alpha_{j3}}e^{\epsilon_{jt}}.
    \end{equation*}
    \end{assumption}

    The Cobb-Douglas technology is a natural choice as it is homogeneous of degree $RTS = \alpha_{j1} + \alpha_{j2} + \alpha_{j3}$. Furthermore, it is easy to see that the output elasticities are given by
    \begin{align*}
        \alpha_{j1} &= \frac{RTS_j w_{jt}^1h_{jt}^1}{w_{jt}h_{jt}^1 + w_{jt}h_{jt}^2 + c_{jt}} \\
        \alpha_{j2} &= \frac{RTS_j w_{jt}^2 h_{jt}^2}{w_{jt}h_{jt}^1 + w_{jt}h_{jt}^2 + c_{jt}} \\
        \alpha_{j3} &= \frac{RTS_j c_{jt}}{w_{jt}h_{jt}^1 + w_{jt}h_{jt}^2 + c_{jt}}.
    \end{align*}

    \noindent
    In words, the model implies that each output elasticity equates a fraction $RTS$ of its share of total children expenditure. These shares are constant in time, regardless of changes in the shadow price of children welfare, $P_{jt}$. The next result warns against ignoring productivity shocks in the model.
    \begin{claim} \label{claim2}
        Suppose Assumptions \ref{Hicksneutrality}-\ref{Cobb-Douglas} hold. If productivity shocks are ignored, then the data may erroneously reject the model at the true return to scale.
    \end{claim}
    
    \begin{proof}
        In what follows, I remove the $j$ subscript from the variables. Suppose the data are rationalized by the model at the true return to scale $RTS^0 \in (0,1]$ and the true children welfare $W_{t} = (h_{t}^{1})^{\alpha_{1}} (h_{t}^{2})^{\alpha_{2}} (c_{t})^{\alpha_{3}}e^{\epsilon_t}$. Suppose now the econometrician ignores productivity shocks and assumes
        \begin{equation*}
            \widetilde{W}_{t} = (h_{t}^{1})^{\alpha_{1}} (h_{t}^{2})^{\alpha_{2}} (c_{t})^{\alpha_{3}}.
        \end{equation*}
    
        \noindent
        Conditional on $RTS^{0}$, the output elasticities are identified. Therefore, children welfare is also identified. From the first-order conditions of the model and Lemma \ref{lemma:RTS}, I have
        \begin{equation*}
            \widetilde{P}_t = \frac{E_t}{RTS^0 \widetilde{W}_{t}}.
        \end{equation*}
    
        \noindent
        Since $\widetilde{W}_t$ is identified, it follows that $\widetilde{P}_t$ is also identified. It is then obvious that the Afriat inequalities
        \begin{equation*} \label{AI:misspecification}
        U_{s}^{i} - U_{t}^{i} \leq \lambda_{t}^{i} \left[ w_{t}^{i} (l_{s}^{i} - l_{t}^{i}) + (q_{s}^{i} - q_{t}^{i}) + \mathcal{P}_{t}^{i}(Q_s - Q_t) + \widetilde{P}_t^{i}(\widetilde{W}_s - \widetilde{W}_t) \right]
    \end{equation*}

    \noindent
    can be rejected by the data even if the data are consistent with the Afriat inequalities under the correct specification of the production function.
    \end{proof}

    Since output elasticities are a function of returns to scale, Claim \ref{claim2} implies that ignoring productivity shocks may lead to inconsistent output elasticities.\footnote{More generally, productivity shocks are useful as they may absorb omitted variables that could otherwise bias the output elasticities.} The next result shows that the first-order conditions of the model have empirical bite under the Cobb-Douglas specification.
    \begin{claim} \label{claim}
        Suppose Assumptions \ref{Hicksneutrality}-\ref{Cobb-Douglas} hold. The first-order conditions of the model are refutable independently of returns to scale.
    \end{claim}

    \begin{proof}
    By Lemma \ref{lemma:RTS}, I have
    \begin{align*}
        P_{jt}W_{jt} = RTS_{j}^{-1}(w_{jt}^{1}h_{jt}^{1} + w_{jt}^{2}h_{jt}^{2} + c_{jt}) \quad \forall t \in \mathcal{T},
    \end{align*}

    \noindent
    where $RTS \in (0,1]$. For the sake of simplicity, suppose there are only two time periods. As such, I have
    %
    \begin{align*}
        \frac{P_{j2} W_{j2}}{P_{j1} W_{j1}} &= \frac{(w_{j2}^{1}h_{j2}^{1} + w_{j2}^{2}h_{j2}^{2} + c_{j2})}{(w_{j1}^{1}h_{j1}^{1} + w_{j1}^{2}h_{j1}^{2} + c_{j1})}.
    \end{align*}

    \noindent
    Since output elasticities are time invariant, it must be that the following set of equations holds
    \begin{align*}
        \frac{P_{j2} W_{j2}}{P_{j1} W_{j1}} &= \frac{w_{j2}^{1}h_{j2}^{1}}{w_{j1}^{1}h_{j1}^{1}} \\
        \frac{P_{j2} W_{j2}}{P_{j1} W_{j1}} &= \frac{w_{j2}^{2}h_{j2}^{2}}{w_{j1}^{2}h_{j1}^{2}} \\
        \frac{P_{j2} W_{j2}}{P_{j1} W_{j1}} &= \frac{c_{j2}}{c_{j1}}.
    \end{align*}

    \noindent
    Note that these equations do not depend on returns to scale. Furthermore, they can easily be violated such as with $w_{j2}^1h_{t2}^1 = 1/2$ and $w_{j2}^2h_{t2}^2 = c_{j2} = 1/4$.
    \end{proof}
    
    Claim \ref{claim} shows that the first-order conditions have meaningful implications that can be tested in the data given a Cobb-Douglas specification. The previous nonparametric results further guarantee that returns to scale and thus the output elasticities are restricted.

    \subsection{Measurement Error}
    
    Claim \ref{claim} shows the the model implies a set of overidentifying restrictions on output elasticities. Hence, any measurement error in the inputs, however small, would lead to the erroneous rejection of the model. It follows that any test of the model that does not address this issue would be dubious in my framework. For this reason, I impose mild centering conditions on measurement error. Let $m_{t}^{x} := x_t - x_{t}^{\star}$ denote the difference between the observed and true value of a variable $x_t$ in period $t$.
    \begin{assumption} \label{Assumption:ME}
        $\mathbb{E}[m_{t}^{x}] = 0$, where $x \in \{h^{1},h^{2},c\}$, $t = 1, 2, \dots, T$.
    \end{assumption}

    Assumption \ref{Assumption:ME} requires that observed inputs be consistent with the true inputs on average in the cross-section. Note that I do not require the distribution of measurement error to be parametric or to be identical over time. An indirect benefit of introducing measurement error in inputs is that I will be able to keep households with missing inputs in the application. Further details about the data are discussed in Section \ref{Section:Data}.

    %

    \section{Testing and Estimation} \label{Section:Estimation}

    This section presents the statistical framework used for testing the model and making inference on the production function.

    \subsection{Testing}
    
    Let $\zeta_j := \{(U_{jt}^{i}, \lambda_{jt}^{i}, \mathcal{P}_{jt}^{i},P_{jt}^{i}, W_{jt}, \bm{\alpha}_{jt}, \omega_{jt}, m_{jt}^{x})_{i \in \{1,2\}, x \in \{h^1,h^2,c\}}\}_{t \in \mathcal{T}} \in Z|\mathcal{D}$ denote the set of household-specific latent variables in the model, where $Z$ denote the support of the latent variables and $\mathcal{D}$ denote the support of the data. The revealed preference characterization along with the moment conditions can be used to define the statistical rationalizability of a panel data set $D := \{D_{j}\}_{j \in \mathcal{N}}$. To this end, write the constraints of the model in the form of moment functions:
    \begin{align*}
        g_{ist}^{U}(D_j,\zeta_j) &:= \bm{\mathbbm{1}} \Bigg( U_{js}^{i} - U_{jt}^{i} \leq \lambda_{jt}^{i} \Big[ w_{jt}^{i} (l_{js}^{i} - l_{jt}^{i}) + (q_{js}^{i} - q_{jt}^{i}) + \\
        &+ \mathcal{P}_{jt}^{i}(Q_{js} - Q_{jt}) + P_{jt}^{i}(W_{js} - W_{jt}) \Big] \Bigg) - 1 \\
        g_{t}^{\alpha_1}(D_j,\zeta_j) &:= \bm{\mathbbm{1}} \Bigg( \alpha_{jt1} = \frac{w_{jt}^{1}h_{jt}^{1}}{P_{jt}W_{jt}} \Bigg) - 1 \\
        g_{t}^{\alpha_2}(D_j,\zeta_j) &:= \bm{\mathbbm{1}} \Bigg( \alpha_{jt2} = \frac{w_{jt}^{2}h_{jt}^{2}}{P_{jt}W_{jt}} \Bigg) - 1 \\
        g_{t}^{\alpha_3}(D_j,\zeta_j) &:= \bm{\mathbbm{1}} \Bigg( \alpha_{jt3} = \frac{c_{jt}}{P_{jt}W_{jt}} \Bigg) - 1 \\
        g_{t}^{W}(D_j,\zeta_j) &:= \bm{\mathbbm{1}} \Bigg( W_{jt} = (h_{jt}^{1})^{\alpha_{j1}}(h_{jt}^{2})^{\alpha_{j2}}(c_{jt})^{\alpha_{j3}}e^{\epsilon_{jt}} \Bigg) - 1 \\
        g_{xt}^{m}(D_j,\zeta_j) &:= m_{jt}^{x},
    \end{align*}

    \noindent
    where $\mathbbm{1}(\cdot)$ denote the indicator function and equates $1$ if the expression inside the parenthesis is satisfied and $0$ otherwise. The latent variables further need to satisfy their support constraints such that $q_{jt}^{1} + q_{jt}^{2} = q_{jt}$ and $\mathcal{P}_{jt}^{1} + \mathcal{P}_{jt}^{2} = 1$. 
    
    Note that I let output elasticities vary in time in the moment functions $g_{t}^{\alpha_k}(D_j,\zeta_j)$. This guarantees that the equations for the output elasticities can be satisfied in every period. To obtain time invariant output elasticities, I require their expected variance to be zero:
    \begin{equation*}
        \mathbb{E}[\bm{g}^{v}(D_j,\zeta_j)] = 0,
    \end{equation*}

    \noindent
    where $\bm{g}^{v}(D_j,\zeta_j) := var(\bm{\alpha})$. Since the variance is always positive, those moment conditions are satisfied if and only if the variance is zero for all households. As such, this formulation is equivalent to directly imposing that production parameters are time invariant.\footnote{This is formally proven in \cite{AK2021}.}
    
    In what follows, I let $\bm{g}(D_j,\zeta_j)$ denote the vector of all moment functions, $\bm{g}^{(m,v)}(D_j,\zeta_j) := (\bm{g}^{m}(D_j,\zeta_j)',\bm{g}^{v}(D_j,\zeta_j)')'$ denote the set of moment functions on measurement error and variance of output elasticities, and $\bm{g}^{-(m,v)}(D_j,\zeta_j) := (\bm{g}^{U}(D_j,\zeta_j)', \bm{g}^{\alpha}(D_j,\zeta_j)', \bm{g}^{W}(D_j,\zeta_j)')'$ denote its complement.
    \begin{definition}
        Under Assumptions \ref{Hicksneutrality}-\ref{Assumption:ME}, a data set $D$ is statistically rationalizable if
        \begin{equation*}
            \underset{ \mu \in \mathcal{M}_{\mathcal{Z}|\mathcal{D}} }{\inf} \Vert \mathbb{E}_{\mu \times \pi_0} [\bm{g}(D,\zeta)] \Vert = 0,
        \end{equation*}
        where $\mathcal{M}_{Z|\mathcal{D}}$ is the set of all conditional probability distributions on $Z|\mathcal{D}$ and $\pi_0 \in \mathcal{M}_{\mathcal{X}}$ is the observed distribution of $D$.
    \end{definition}

    In its current form, the notion of statistical rationalizability has $2T^2 + T + T + T + T + 3T + T$ moment conditions, including some that are discontinuous. Let $d_m$ denote the number of moment conditions on measurement error and $d_v$ denote the number of moment conditions on the variance of the output elasticities. The following result due to \cite{Schennach2014} and \cite{AK2021} allows us to considerably reduce the complexity of the problem.
    \begin{proposition} \label{propositionELVIS}
    Under Assumptions \ref{Hicksneutrality}-\ref{Assumption:ME}, a data set $D$ is statistically rationalizable if and only if 
    \begin{equation*}
        \underset{\bm{\gamma} \in \mathbb{R}^{d_m + d_v}}{\min} \Vert \mathbb{E}_{\pi_0}[\bar{\bm{g}}(D;\bm{\gamma})] \Vert = 0,
    \end{equation*}
    where 
    {\footnotesize
    \begin{equation*}
        \bar{g}_j(D_j;\bm{\gamma}) := \frac{ \int_{\zeta_j \in Z|\mathcal{D}} \bm{g}_{j}^{(m,v)}(D_j,\zeta_j) \exp(\bm{\gamma}' \bm{g}_{j}^{(m,v)}(D_j,\zeta_j)) \bm{\mathbbm{1}}(\bm{g}_{j}^{-(m,v)}(D_j,\zeta_j) = 0) \,d\eta(\zeta_j|D_j) }{ \int_{\zeta_j \in Z|\mathcal{D}} \exp(\bm{\gamma}' \bm{g}_{j}^{(m,v)}(D_j,\zeta_j)) \bm{\mathbbm{1}}( \bm{g}_{j}^{-(m,v)}(D_j,\zeta_j) = 0) \,d\eta(\zeta_j|D_j)},
    \end{equation*}
    }
    \noindent
    and $\eta(\cdot|D_j)$ is an arbitrary user-specified distribution supported on $Z|\mathcal{D}$ such that $\mathbb{E}_{\pi_{0}}[ \log(\mathbb{E}_{\eta}[\exp(\bm{\gamma}' \bm{g}^{(m,v)}(D,\zeta))|D])]$ exists and is twice continuously differentiable in $\bm{\gamma}$ for all $\bm{\gamma} \in \mathbb{R}^{d_m+d_v}$.
    \end{proposition}


    The previous result calls for some comments. First, the dimensionality of the problem is greatly reduced as it only requires finding a finite dimensional parameter $\bm{\gamma}$ rather than a distribution $\mu$. Second, the moment conditions associated with the concavity of the utility functions, first-order conditions, and production function equations are directly imposed on each household data set such as to restrict the support of the unobservables. In particular, observe that the optimization problem no longer includes any discontinuous moment condition. Finally, it is worth noting that the result states that there is no loss in generality in averaging out the unobservables in the moment functions provided the distribution is from the exponential family.

    The simplification allowed by Proposition \ref{propositionELVIS} requires finding unobservables $\zeta_j$ that exactly satisfy the concavity of the utility functions, first-order conditions, and production function equations. If the constraints were linear in the unobservables, it would be possible to use a standard Hit-and-Run algorithm to directly sample them from the feasible space defined by the intersection of the inequalities and the system of equations. Unfortunately, the inequalities are highly nonlinear, therefore making this approach impossible.\footnote{In principle, it would be possible to use rejection sampling along with a mixed-integer programming (MIP) problem to draw from the feasible space. However, these types of MIP for collective models are NP-complete (\citeauthor{Nobibon2016}, \citeyear{Nobibon2016}) so they do not scale well. Also, rejection sampling is generally slow.}

    I resolve this pervasive issue by proposing a blocked Gibbs sampler. The idea is to break down the sampling procedure into multiple blocks, where each block takes a subset of all unobservables as given. The key is to create those blocks in such a way that the inequalities are linear in the unobservables conditional on a certain subset of all unobservables. Thus, the inequalities effectively define a (conditional) convex polytope in each block. This allows for a straightforward sampling procedure that guarantees the unobservables to exactly satisfy the inequalities, first-order conditions, and production function equations. The details of the algorithm are provided in the Appendix.

    \subsection{Inference}

    One of the advantages of ELVIS is that testing and inference are quite simple even if the model is partially identified. Indeed, testing the model can be done by constructing the sample analogues of the averaged moments and by computing a test statistic that is stochastically bounded by the chi-square distribution. Inference is achieved by further adding moment conditions on parameters of interest and inverting the test statistic. Since the test statistic is stochastically bounded by the chi-square distribution, it suffices to compare the value of the test statistic against the chi-square critical value with $d_m + d_v$ ($d_m + d_v + d_\theta$) degrees of freedom for testing (inference). Importantly, the identified set is convex under mild conditions.\footnote{I refer the reader to \cite{AK2021} for additional details about the statistical procedure.}

    \section{Data} \label{Section:Data}

    I conduct my empirical analysis with the Longitudinal Internet Studies for the Social Sciences (LISS) panel data. The panel consists of about $5000$ households representative of the Dutch population and gathers information about panelists yearly. Since the LISS data directly include information on private expenditures within the household, an important point of departure from the model is that private expenditures $q_{t}^{i}$ are observed.
    
    The time use data were collected by means of survey questions about the time spent on a set of time use categories during the past seven days. Although the survey is not demanding of household members memory, the actual time allocations throughout the month are likely to differ from the ones reported at the time of the survey. Similarly, data on monthly expenditures were collected via survey questions. Additional details relating to data collection can be found in \cite{Cherchye2012}.
    
    
    Since my main goal is to estimate the production function, I only consider measurement error in inputs. Nevertheless, I observe that my methodology could accommodate measurement error in other variables. This choice also explicitly recognizes that the overidentifying restrictions implied by the Cobb-Douglas specification may not hold perfectly due to the presence of measurement error. Indeed, there is empirical evidence documenting that mean expenditures are correctly reported in survey data \citep{KolsrudME2017,AbildgrenME2018}, which motivates the restriction of mean zero measurement error. The reader is referred to Section \ref{section:specification} for details on the specification of the production function and the restrictions on measurement error.

    My empirical analysis focuses on couples with children. This restriction alone reduces the number of households to about a thousand. I further restrict the set of observations with nonmissing and nonzero wages, private expenditures, and public expenditures. For households with missing or zero data on inputs, I impute their values. Note that the treatment of measurement error explicitly handles the imperfection of the imputation. Lastly, I restrict the sample to households that are in the panel for three periods.
    
    The final sample consists of $132$ couples with children observed over $3$ time periods pooled from the years $2008$ to $2017$.\footnote{Ideally, it would be desirable for the sample size to be larger for the asymptotic theory to fully apply. This limitation of the data set warrants future efforts to increase the sample size such as by including households with missing wages.} While the sample size is relatively small, I note that it is comparable with \cite{Cherchye2012} despite my restriction to panel data. Summary statistics of the sample are displayed in Table \ref{tab:data}. Further details about the sample construction are given in the Appendix.


\begin{center}
    \captionof{table}{Sample Summary Statistics} \label{tab:data}
    \begin{adjustbox}{width=1.0 \textwidth,center=\textwidth}
       \begin{threeparttable}[b]
    \begin{tabular}{ p{2mm} l c c c c c c}
        \toprule
         & &  \multicolumn{2}{c}{Husband} & \multicolumn{2}{c}{Wife} & \multicolumn{2}{c}{Household} \\
        & & Mean & Std dev. & Mean & Std dev. & Mean & Std dev.\\
        \hline
        \multicolumn{2}{l}{Age} & 46.25 & 7.99 & 44.08 & 7.28 \\
        \multicolumn{2}{l}{Wage (EUR/hour)} &  13.68 & 6.65 & 12.95 & 12.00 \\
        \multicolumn{2}{l}{Number of children} & & & & & 2.00 & 0.83 \\
        \multicolumn{2}{l}{Mean age of children} & & & & & 13.15 & 6.36 \\
        \multicolumn{2}{l}{Childcare (hours/week)} & 10.41 & 7.89 & 16.92 & 13.14 \\
        \multicolumn{2}{l}{Work (hours/week)} & 37.50 & 6.03 & 23.30 & 8.27 \\
        \multicolumn{2}{l}{Private expenditure (EUR/month)} & 362.00 & 229.23 & 395.81 & 323.30 \\
        \multicolumn{2}{l}{Public expenditure (EUR/month)} &  &  & & & 2164.85 &  836.12 \\
        \multicolumn{2}{l}{Children expenditure (EUR/month)} &  &  & & & 512.52 & 512.90 \\
        \toprule
        \multicolumn{2}{l}{Total households} & & & & & & 132 \\
        \bottomrule 
    \end{tabular}
      \end{threeparttable}
    \end{adjustbox}
\end{center}

    \section{Empirical Results} \label{Section:Results}

    This section recovers confidence sets on expected returns to scale and expected output elasticities. Then, it investigate how the production function changes with demographics.

    \subsection{Expected Parameters}
    
    I begin the empirical analysis by recovering the $95\%$ confidence set on expected returns to scale. Since the confidence set is convex, I only need to find the lower and upper bound on the expected return to scale. I find that the $95\%$ confidence set on expected returns to scale is $[0.270,0.405]$. Since the confidence set is nonempty, it follows that the model is not rejected by the data. Hence, I continue the analysis and recover $95\%$ confidence sets on expected output elasticities. The results are reported in Figure \ref{fig:outputelasticities}.
    \begin{figure}[h]
        \centering
	  \includegraphics[width=12cm, trim=10 10 10 10, clip]{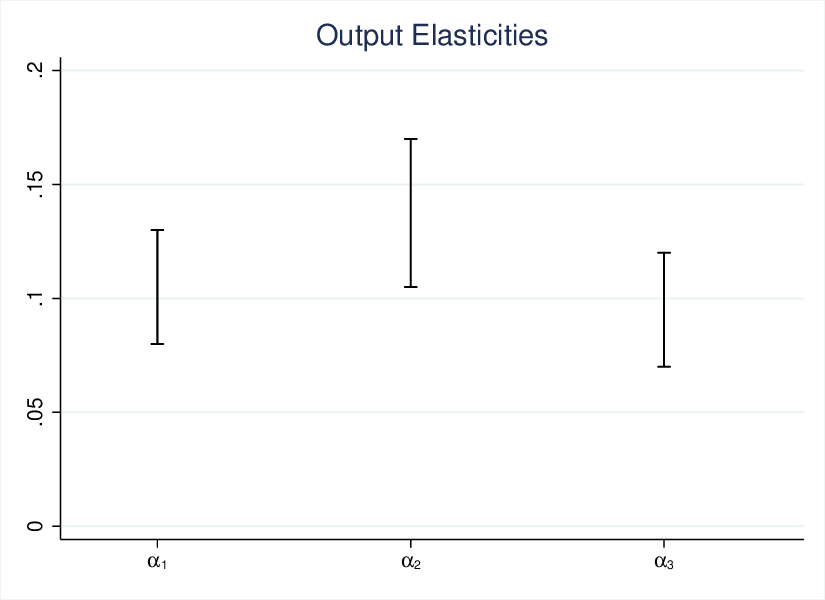}
	  \caption{$95\%$ Confidence Sets on Production Parameters} \label{fig:outputelasticities}
    \end{figure}

    The figure show that, on average, time inputs by mothers increases children welfare by more than time inputs by fathers or children expenditure. More precisely, a $10\%$ increase in childcare by the mother increases children welfare by about $14\%$ while a $10\%$ increase in childcare by the father increases children welfare by about $10\%$. The impacts of children expenditure appears lowest with a $10\%$ increase in children expenditure increasing children welfare by about $9\%$.

    \subsection{Heterogeneity}


    
    This section investigates how the expected production technology varies with household characteristics. Due to the small sample size, I choose to the analysis using linear regressions. First, I analyze heterogeneity in returns to scale through the regression
    \begin{align} \label{eq:reg}
        \bm{RTS} &= \bm{X}\bm{\beta} + \bm{\omega},
    \end{align}

    \noindent
    where $\bm{RTS} = \bm{\alpha}_1 + \bm{\alpha}_2 + \bm{\alpha}_3$ represents returns to scale, $\bm{X}$ is a set of covariates, and $\bm{\omega}$ is a random error. Likewise, I analyze heterogeneity in output elasticities through the regressions
    \begin{align}
        \bm{\alpha}_k &= \bm{X}\widetilde{\bm{\beta}}_k + \widetilde{\bm{\omega}}_k,
    \end{align}
    
    \noindent
    where $\bm{\alpha}_k$ represents output elasticities with respect to input $k$, $\bm{X}$ is a set of covariates, and $\widetilde{\bm{\omega}}_k$ is a random error. I assume the errors are uncorrelated and mean zero conditional on the data.\footnote{In the results below, I impose and test $\mathbb{E}[X_k \bm{\omega}] = 0$, where $k$ denote the $k$th covariate in the regression. I view the nonrejection of the restriction as support for the conditional mean zero assumption.} I wish to emphasize that these regression equations are not estimated separately from the rest of the model, but rather imposed as additional equation restrictions within the model. That is, I make inference on the expected parameters of a regression by adding a moment function such as
    \begin{align*}
        g^{\beta}(D_j,\zeta_j) := \bm{\mathbbm{1}}(RTS_j = \bm{X}_j\bm{\beta} + \omega_j).
    \end{align*}

    \noindent
    The $95\%$ confidence sets on the expected coefficients of the regressions are reported in Table \ref{tab:beta}.
    \begin{table}[htbp]
        \captionof{table}{$95\%$ Confidence Sets on Regression Coefficients} \label{tab:beta}
        \centering
        \begin{adjustbox}{width=1.0 \textwidth,center=\textwidth}
            \begin{threeparttable}[b]
                \begin{tabular}{c c c c c}
                \hline
                \multicolumn{1}{c}{} & \multicolumn{4}{c}{Dependent Variable} \\
                Independent Variable\tnote{a} & $\alpha_1 + \alpha_2 + \alpha_3$ & $\alpha_1$ & $\alpha_2$ & $\alpha_3$ \\ \hline 
                havo, vwo \& mbo (Father) & $[-0.08, 0.05]$ & $[0.02, 0.04]$ & $[-0.02,0.00]$ & $[-0.09,0.03]$ \\
                hbo \& wo (Father) & $[-0.12, 0.03]$ & $[0.00,0.02]$ & $[0.04, 0.07]$ & $[-0.10, 0.05]$ \\
                havo, vwo \& mbo (Mother) & $[0.06, 0.20]$ & $[0.02, 0.05]$ & $[0.05, 0.08]$ & $[-0.05, 0.13]$ \\
                hbo \& wo (Mother) & $[0.02, 0.15]$ & $[0.01, 0.03]$ & $[0.04, 0.07]$ & $[-0.075, 0.10]$ \\
                \#Children & $[-0.06, -0.01]$ & $[-0.03,0.02]$ & $[-0.03, -0.01]$ & $[-0.03, 0.04]$ \\
                Age Children & $[0.00, 0.01]$ & $[-0.03, 0.04]$ & $[0.00, 0.00]$ & $[0.00, 0.01]$ \\
                Dwelling & $[-0.35, -0.10]$ & $[-0.08, -0.06]$ & $[-0.12, -0.10]$ & $[-0.15, 0.10]$ \\
                \hline
                \end{tabular}
                \begin{tablenotes}
                \item[a] \setstretch{1.0} \footnotesize{The education category ``havo, vwo, \& mbo" represents general education that leads to higher education and vocational education that can lead to higher education. The education category ``hbo \& wo" represents higher education. \#children is the number of children in the household. Age Children is the average age of children in the household. Dwelling is an indicator that takes value $0$ if the household rents and $1$ if it owns a house.}
                \end{tablenotes}
            \end{threeparttable}
        \end{adjustbox}
    \end{table}

    The first four rows of Table \ref{tab:beta} capture the impacts of education on the production technology, where education is a categorical variable that reflects the type of education in the Netherlands. The first category represents primary school and pre-vocational secondary education (VMBO) and is used as the base in the regression. The second category represents general education that leads to higher education (HAVO and VWO) and vocational education (MBO) that can lead to higher education. The third category represents higher education (HBO, WO).

    Table \ref{tab:beta} shows that a higher education level increases the expected output elasticity by about $0.02$ for fathers and $0.06$ for mothers. Further, there is a positive spillover effect whereby higher education of one household member increases the expected output elasticity of the other member, with the exception of fathers with havo, vwo \& mbo education. Interestingly, the education level of fathers does not appear to increase expected returns to scale. In contrast, the education level of mothers increases expected returns to scale by about $0.10$. Finally, I cannot infer the impacts of education on the expected output elasticity with respect to children expenditure as the confidence sets are large and include zero.

    Table \ref{tab:beta} also shows that the number of children in the household has a small but negative impact on expected returns to scale and expected output elasticities with respect to time inputs. More interestingly, I find that the dwelling of the household has a large negative impacts on the production of children welfare, where the results are relative to owning a house. For example, households that rent have expected output elasticities with respect to time inputs lower by about $0.07$ and $0.11$ for fathers and mothers, respectively. Further, with $95\%$ confidence, expected returns to scale decrease by at least $0.10$ up to $0.35$ for households that do not own a house.


    \section{Conclusion} \label{Section:Conclusion}

    This paper proposes a novel framework to assess the impacts of welfare reforms directed at families. I show that conditions for point identification in collective models with children such as constant returns to scale are rejected by the data. Instead, my findings are consistent with the empirical human capital literature, such as in \cite{delBoca2014}, where the impacts of parental inputs on children development exhibit decreasing returns to scale. Although my results provide support for the collective model approach to analyzing children development, they also warn against assuming constant returns to scale in the production technology. Indeed, children human capital development is now recognized as an important factor in assessing the costs and benefits of welfare reforms \citep{Mullins2022}, but such cost-benefit analysis crucially depends on the shape of the production technology. Furthermore, my empirical results show that the education level of the mother and the household environment play important roles in the development of children. Those findings provide support for policy interventions targeted at disadvantaged households such as to mitigate children achievement gaps. In spite of the advances made in this paper, there is a need for future research along multiple directions. First, an analysis of children development that allows for potential complementarities in inputs within the collective model may reveal additional interesting patterns in the production of children human capital. Second, additional data on time use such as active and passive time spent with children may provide insights into the reasons for differences in parenting skills between household demographics. Third, an extension of the collective model with children to a dynamic setup would enable one to quantify how early investments in children impact returns to later investments.

    \newpage

    \noindent
    {\LARGE \textbf{Appendix}}
    
    \appendix    

    \section{Sample Construction}

    For each household, I compute how many children live at home and focus on households with children living at home. I drop observations that pertain to single households and those that do not have any child living at home. Next, I remove observations for which one or both members have zero or missing wages. In this way, I avoid erroneously assuming a household member works when (s)he is not and vice versa. Furthermore, I remove observations where hours worked is missing to improve the quality of the imputation of leisure. Finally, I remove observations for which there is zero or missing private or public expenditures, as well as observations for which private or public expenditures are greater than the household gross income.

    I impute zero and missing values of childcare and children expenditure for observations that remain in the sample after the previous selection criteria. The imputation is a simple year average of the variable. This imputation is likely to be an overestimate of the actual value for some households and an underestimate for others. Hence, it should be consistent with the moment conditions on measurement error.

    Besides the imputation for some zero or missing inputs, I compute leisure of each household member as a residual according to the following equation:
        \begin{align*}
        l_{t}^{i} &= 168 - 56 - b_{t}^{i} - h_{t}^{i},
    \end{align*}

    \noindent
    where $l_{t}^{i}$ is leisure, $168$ is the total number of hours in a week, $56$ is the number of hours spent sleeping in a week, $b_{t}^{i}$ is time spent working, and $h_{t}^{i}$ is time spent on childcare. I drop households where implied leisure is negative as this may reflect a major problem with time data. There was only one household where this occurred in the sample.  Clearly, my construction of leisure may still be inaccurate. I tackle this problem by allowing for measurement error in leisure. Precisely, since the time constraint requires $l_{t}^{i} + b_{t}^{i} + h_{t}^{i} = 168 - 56$ and $h_{t}^{i}$ is mismeasured, I let true leisure be minus true childcare ($l_{t}^{\star i} = - h_{t}^{\star i}$) such that the time constraint holds at the true variables. These weekly variables are then scaled such as to obtain time inputs for the average number of days in a month. Since there are seven days in a week and a month has slightly more than $30$ days on average, I multiply time inputs by $4.3$.
    
    As a last refinement of the sample, I remove households that are part of the LISS data for strictly less than 3 years. To obtain a balanced panel, I keep the first $3$ observations of each household that is present for strictly more than $3$ years. Thus, my sample is composed of households from various sets of $3$ periods (e.g., $2009$-$2010$-$2012$ or $2010$-$2012$-$2015$). I limit myself to a three-year panel despite the greater empirical bite that could be obtained with additional periods to avoid any additional decrease in the sample size. Lastly, I remove households with missing demographic information as it is necessary for the linear regression in Section \ref{Section:Results}.

    \section{Proofs}
    
    \subsection{Proof of Theorem 1}

    \noindent
    $(i) \implies (ii)$ 

    \noindent
    The household problem can be written as
    \begin{align*}
        \max_{ (l^1,l^2,h^1,h^2,q^{1},q^{2},Q,c) \in \mathbb{R}_{+}^{2} \times \mathbb{R}_{++}^{2} \times \mathbb{R}_{+}^{2 L}\times \mathbb{R}_{+} \times \mathbb{R}_{++}}  \mu_{t}^{1} U^{1}(l^1,q^1,Q,W) +
        \mu_{t}^{2} U^{2}(l^2,q^{2},Q,W),
    \end{align*}

	\noindent
	subject to satisfying the household constraints
 	\begin{align*}
	    (q^{1} + q^{2}) + Q + c &= w_{t}^{1} (\tau-l^1-h^1) + w_{t}^{2}(\tau-l^2-h^2) \\
            W &= F(h^1,h^2,c)e^{\epsilon_t}.
	\end{align*}

    \noindent
    The first-order conditions are given by
    \begin{align*}
        \mu_{t}^{i} \pdv{{U}^{i}}{l^i} &= \eta_t w_{t}^{i} \\
        \mu_{t}^{i} \pdv{{U}^{i}}{q^i} &= \eta_t. \\
        \sum_{i} \mu_{t}^{i} \pdv{{U}^{i}}{W_t} \cdot \pdv{W_t}{h^{1}} &= \eta_t w_{t}^{1} \\
        \sum_{i} \mu_{t}^{i} \pdv{{U}^{i}}{W_t} \cdot \pdv{W_t}{h^{2}} &= \eta_t w_{t}^{2} \\
        \sum_{i} \mu_{t}^{i} \pdv{{U}^{i}}{W_t} \cdot \pdv{W_t}{c} &= \eta_t \\
        \sum_{i} \mu_{t}^{i} \pdv{{U}^{i}}{Q} &= \eta_t,
    \end{align*}

    \noindent
    where the equalities hold for some supergradient of the utility function.\footnote{For corner solutions, the first-order conditions may only hold with inequality. The argument does not require any substantive change to accommodate this possibility.} Next, define
    \begin{align*}
        \lambda_{t}^{i} &= \frac{\eta_t}{\mu_{t}^{i}} \\
        P_t^{i} &= \frac{\mu_{t}^{i}}{\eta_t} \pdv{U^i}{W_t} \\
        \mathcal{P}_{t}^{i} &= \frac{\mu_{t}^{i}}{\eta_t} \pdv{U^i}{Q}.
    \end{align*}

    \noindent
    The first-order conditions can be rewritten as
    \begin{align} \label{convenientFOCs0}
        \pdv{{U}^{i}}{l^i} &= \lambda_{t}^{i} w_{t}^{i} \\ \label{convenientFOCs1}
        \pdv{{U}^{i}}{q^i} &= \lambda_{t}^{i}. \\ \label{convenientFOCs2}
        (P_t^{1} + P_t^{2}) \pdv{W_t}{h^1} &= w_{t}^{1} \\ \label{convenientFOCs3}
        (P_t^{1} + P_t^{2}) \pdv{W_t}{h^2} &=  w_{t}^{2} \\ \label{convenientFOCs4}
        (P_t^{1} + P_t^{2}) \pdv{W_t}{c} &= 1 \\
        \mathcal{P}_{t}^{1} + \mathcal{P}_{t}^{2} &= 1 \label{eq:lindahl}.
    \end{align}

    \noindent
    Using the concavity of the utility functions I obtain
    \begin{align*}
        U_{s}^{i} - U_{t}^{i} &\leq \left[ \pdv{{U}^{i}}{l_{t}^i} \: (l_{s}^{i} - l_{t}^{i}) + \pdv{{U}^{i}}{q_{t}^i} \: (q_{s}^{i} - q_{t}^{i}) + \pdv{{U}^{i}}{Q} \:(Q_s - Q_t) + \pdv{{U}^{i}}{W_t} \: (W_s - W_t) \right],
    \end{align*}

    \noindent
    where $U_{t}^{i} := U^{i}(l_{t}^{i},q_{t}^{i},Q_t,W_t)$ for all $t \in \mathcal{T}$. Substituting the derivatives of the utility function for their expressions yields
    \begin{align*}
        U_{s}^{i} - U_{t}^{i} &\leq \lambda_{t}^{i} \Bigl[ w_{t}^{i} (l_{s}^{i} - l_{t}^{i}) + (q_{s}^{i} - q_{t}^{i}) + \mathcal{P}_{t}^{i} (Q_s - Q_t) + P_t^{i} (W_s - W_t)  \Bigr].
    \end{align*}

    \noindent
    Next, using the concavity of the production function I obtain
    \begin{align*}
        F_{s} - F_{t} \leq \pdv{{F}}{h_{t}^1} \: (h_{s}^{1} - h_{t}^{1}) + \pdv{{F}}{h_{t}^2} \: (h_{s}^{2} - h_{t}^{2}) + \pdv{{F}}{c_{t}} \: (c_s - c_t),
    \end{align*}

    \noindent
    where $F_{t} := F(h_{t}^{1},h_{t}^{2},c_t)$ for all $t \in \mathcal{T}$. Substituting the derivatives of the production function from equations \eqref{convenientFOCs2}-\eqref{convenientFOCs4} yields
    \begin{align*}
        F_{s} - F_{t} &\leq \frac{w_{t}^{1}}{P_te^{\epsilon_t}} (h_{s}^{1} - h_{t}^{1}) + \frac{w_{t}^{2}}{P_te^{\epsilon_t}} (h_{s}^{2} - h_{t}^{2}) + \frac{1}{P_te^{\epsilon_t}} (c_s - c_t).
    \end{align*}
    
    \noindent
    Putting everything together, these inequalities should hold for some $U_{t}^{i}$, $\lambda_{t}^{i} > 0$, $\mathcal{P}_{t}^{i} > 0$ such that $\mathcal{P}_{t}^{1} + \mathcal{P}_{t}^{2} = 1$, $P_t^{i} > 0$ such that $P_t^1 + P_t^2 = P_t$, $W_t$, $F_t > 0$ and $\epsilon_t$ such that $W_t = F_te^{\epsilon_t}$, $ t = 1, \dots, T$.

    \noindent
    $(ii) \implies (i)$\\
    \noindent
    I have to show that, if Theorem \ref{proposition1} $(ii)$ holds, then there exist concave utility functions and a concave production function that rationalize the data. Thus, let $\tau = \{t_j\}_{j=1}^{m}$, $m \geq 2$, $t_j \in \mathcal{T}$ denote a sequence of indices and $\mathcal{I}$ denote the set of all such indices. Define
    \begin{align*}
        U^{i}&(l^i,q^{i},Q,W) := \\
        &\underset{\tau \in \mathcal{I}}{\min} \Bigl\{ \lambda_{t_m}^{i} \Bigl[ w_{t_m}^{i} (l^{i} - l_{t_m}^{i}) + (q^{i} - q_{t_m}^{i}) + \mathcal{P}_{t_m}^{i} (Q - Q_{t_m}) + P_{t_m}^{i} ( W - W_{t_m} ) \Bigr] + \\
        &+ \sum_{j=1}^{m-1} \lambda_{t_j}^{i} \Bigl[ w_{t_j}^{i} (l_{t_{j+1}}^{i} - l_{t_j}^{i}) + (q_{t_{j+1}}^{i} - q_{t_j}^{i}) + \mathcal{P}_{t_j}^{i} (Q_{t_{j+1}} - Q_{t_j}) + P_{t_j}^{i} ( W_{t_{j+1}} - W_{t_j} ) \Bigr] 
        \Bigr\}.
    \end{align*}
    
    \noindent
    The function is the pointwise minimum of a collection of linear functions. Thus, it is continuous, increasing, and concave. By definition of $U^{i}$, there is some sequence of indices such that
    \begin{align*}
        U^{i}&(l_{t}^i,q_{t}^{i},Q_t,W_t) \geq \\
        &\lambda_{t_m}^{i} \Bigl[ w_{t_m}^{i} (l_{t}^{i} - l_{t_m}^{i}) + (q_{t}^{i} - q_{t_m}^{i}) + \mathcal{P}_{t_m}^{i} (Q_t - Q_{t_m}) + P_{t_m}^{i} ( W_t - W_{t_m}) \Bigr] + \\
        &+ \sum_{j=1}^{m-1} \lambda_{t_j}^{i} \Bigl[ w_{t_j}^{i} (l_{t_{j+1}}^{i} - l_{t_j}^{i}) + (q_{t_{j+1}}^{i} - q_{t_j}^{i}) + \mathcal{P}_{t_j}^{i} (Q_{t_{j+1}} - Q_{t_j}) + P_{t_j}^{i} (W_{t_{j+1}} - W_{t_j})  \Bigr].
    \end{align*}

    \noindent
    Add any allocation $(l^i,q^{i},Q,W)$ to the sequence and use the definition of $U^i$ once again to obtain
    \begin{align*}
        &\lambda_{t}^{i} \Bigl[ w_{t}^{i} (l^{i} - l_{t}^{i}) + (q^{i} - q_{t}^{i}) + \mathcal{P}_{t}^{i} (Q - Q_{t}) + P_t^{i} (W - W_t) \Bigr] + \\
        &+ \lambda_{t_m}^{i} \Bigl[ w_{t_m}^{i} (l_{t}^{i} - l_{t_m}^{i}) + (q_{t}^{i} - q_{t_m}^{i}) + \mathcal{P}_{t_m}^{i} (Q_t - Q_{t_m}) + P_{t_m}^{i} (W_t - W_{t_m}) \Bigr] + \\
        &+ \sum_{j=1}^{m-1} \lambda_{t_j}^{i} \Bigl[ w_{t_j}^{i} (l_{t_{j+1}}^{i} - l_{t_j}^{i}) + (q_{t_{j+1}}^{i} - q_{t_j}^{i}) + \mathcal{P}_{t_j}^{i} (Q_{j+1} - Q_{t_j}) + P_{t_j}^{i}  ( W_{t_{j+1}} -  W_{t_j} ) \Bigr] \\
        &\geq U^{i}(l^i,q^{i},W,Q).
    \end{align*}

    \noindent
    Hence, rearranging the previous expression yields
    \begin{align*}
        U^{i}&(l^i,q^{i},Q,W) - U^{i}(l_{t}^i,q_{t}^{i},Q_t,W_t) \leq \lambda_{t}^{i} \Bigl[ w_{t}^{i} (l^{i} - l_{t}^{i}) + (q^{i} - q_{t}^{i}) + \mathcal{P}_{t}^{i} (Q - Q_{t}) + P_t^{i} ( W - W_t ) \Bigr].
    \end{align*}

    \noindent
    Note that the two first supergradients of $U^{i}(l_t^i,q_{t}^{i},Q_t,W_t)$ give the first-order conditions \eqref{convenientFOCs0}-\eqref{convenientFOCs1}. Next, define 
    \begin{align*}
        F&(h^1,h^{2},c) :=
        \underset{\tau \in \mathcal{I}}{\min} \Bigl\{ \frac{1}{P_te^{\epsilon_t}} \Bigl[ w_{t}^{1} (h^{1} - h_{t}^{1}) + w_{t}^{2} (h^2 - h_{t}^{2}) + (c - c_t) + \\
        &+ \sum_{j=1}^{m-1} \Bigl[ \frac{1}{P_{t_j}e^{\epsilon_{t_j}}} \Bigl[ w_{t_j}^{1} (h_{t_{j+1}}^{1} - h_{t_j}^{1}) + w_{t_j}^{2} (h_{j+1}^2 - h_{t_j}^{2}) + (c_{t_{j+1}} - c_{t_j}) \Bigr]
        \Bigr\}.
    \end{align*}

    \noindent
    This function is continuous, increasing, and concave in $(h^1,h^2,c)$. By an identical argument as before, I obtain
    \begin{align*}
        F&(h^1,h^{2},c) - F(h_{t}^1,h_{t}^{2},c_t) \leq \frac{1}{P_te^{\epsilon_t}} \Bigl[ w_{t}^{1} (h^{1} - h_{t}^{1}) + w_{t}^{2}(h^{2} - h_{t}^{2}) + P_t^{i} (c - c_t) \Bigr].
    \end{align*}
    
    \noindent
    Hence, the supergradients of $F(h_t^1,h_t^2,c_t)$ yield equations \eqref{convenientFOCs2}-\eqref{convenientFOCs4} and I conclude that Theorem \ref{proposition1} $(ii)$ has the same implications as the household problem \eqref{model}.

    \noindent
    $(ii) \implies (iii)$
    \noindent

    \noindent
    Let us begin by noting that the Afriat inequalities can be combined such that for all $\{t_k\}_{k=1}^{m} \in \mathcal{I}$ and all $i \in \{1,2\}$
    \begin{align*}
        0 &\leq \sum_{k=1}^{m} \lambda_{t_{k+1}}^{i} a_{t_k,t_{k+1}}^{i}.
    \end{align*}

    \noindent
    Observe that the set of all sequences $\mathcal{I}$ can be reduced to the set of all finite sequences as any sequence that satisfies this inequality is also satisfied without cycles. For the sake of a contradiction, suppose GARP is not satisfied for some household member. Then, there exists a cycle such that $a_{t_1,t_2}^{i} \leq 0$, $a_{t_2,t_3}^{i} \leq 0$, \dots, $a_{t_m,t_1}^{i} < 0$. Thus, it follows that
    \begin{equation*}
        \lambda_{t_2}^{i}a_{t_1,t_2}^{i} + \lambda_{t_3}^{i}a_{t_2,t_3}^{i} + \dots + \lambda_{t_1}^{i}a_{t_m,t_1}^{i} < 0,
    \end{equation*}

    \noindent
    a contradiction of cyclical monotonicity. Next, wish to show that GAPM holds. Observe that the inequalities for the production function can be rearranged as
    \begin{align*}
         P_tF_{t}e^{\epsilon_t} + w_{t}^{1} h_{t}^{1} + w_{t}^{2} h_{t}^{2} + c_t \leq P_tF_{s}e^{\epsilon_t} + w_{t}^{1} h_{s}^{1}  + w_{t}^{2} h_{s}^{2} + c_s \quad \forall s,t \in \mathcal{T},
    \end{align*}

    \noindent
    where I further have $W_t = F_te^{\epsilon_t}$ by assumption.

    \noindent
    $(iii) \implies (ii)$

    \noindent
    Suppose that GARP holds for each household member. Then, an application of \cite{FST2004} shows the existence of the Afriat inequalities for each household member. Furthermore, rearranging the inequalities in GAPM yields the desired inequalities.

    \subsection{Proof of Lemma \ref{lemma:RTS}}

    \begin{proof}
    From the first-order conditions of the model and the Hicks-neutrality of productivity shocks, I have
    \begin{align*}
        \pdv{F(h_{t}^{1},h_{t}^{2},c_t)}{h_{t}^{1}}e^{\epsilon_t} &= \frac{w_{t}^{1}}{P_{t}} \\
        \pdv{F(h_{t}^{1},h_{t}^{2},c_t)}{h_{t}^2}e^{\epsilon_t} &= \frac{w_{t}^{2}}{P_t} \\ 
        \pdv{F(h_{t}^{1},h_{t}^{2},c_t)}{c_{t}^2}e^{\epsilon_t} &= \frac{1}{P_t}.
    \end{align*}

    \noindent
    I can multiply each marginal product by its own factor of production to get
    \begin{align*}
        \pdv{F(h_{t}^{1},h_{t}^{2},c_t)}{h_{t}^{1}}h_{t}^{1}e^{\epsilon_t} &= \frac{w_{t}^{1}h_{t}^{1}}{P_{t}} \\
        \pdv{F(h_{t}^{1},h_{t}^{2},c_t)}{h_{t}^2}h_{t}^{2}e^{\epsilon_t} &= \frac{w_{t}^{2}h_{t}^{2}}{P_t} \\ 
        \pdv{F(h_{t}^{1},h_{t}^{2},c_t)}{c_{t}^2}c_te^{\epsilon_t} &= \frac{c_t}{P_t}.
    \end{align*}

    \noindent
    Summing up these equations and multiplying by $P_t$, I obtain
    \begin{equation*}
        P_t \left[ \pdv{F(h_{t}^{1},h_{t}^{2},c_t)}{h_{t}^{1}}h_{t}^{1} + \pdv{F(h_{t}^{1},h_{t}^{2},c_t)}{h_{t}^2}h_{t}^{2} + \pdv{F(h_{t}^{1},h_{t}^{2},c_t)}{c_{t}^2}c_t \right]e^{\epsilon_t} = E_t,
    \end{equation*}

    \noindent
    where $E_t := w_{t}^{1}h_{t}^{1} + w_{t}^{2}h_{t}^{2} + c_t$. Since the production function is homogeneous of degree $RTS \in (0,1]$, an application of Euler's theorem gives
    \begin{equation*}
        RTS P_tW_t = E_t,
    \end{equation*}

    \noindent
    where I used the production function equation $W_t = F(h_{t}^{1},h_{t}^{2},c_t)e^{\epsilon_t}$.
    \end{proof}

    \section{Sampling from the Feasible Space: A Blocked Gibbs Sampler}

    This section explains how to draw latent variables that satisfy the household problem. Since private expenditures are directly observed in the application, I do not need to find such quantities in the procedure. It is quite straightforward to extend the procedure to further find private expenditures if those were not observed, however.
    
    Let $P_t = P_t^{1} + P_t^{2}$ and recall that the data are consistent with the model if there exist personalized prices $P_t^{i} > 0$, personalized prices $\mathcal{P}_{t}^{i} > 0$ such that $\mathcal{P}_{t}^{1} + \mathcal{P}_{t}^{2} = 1$, numbers $U^{i}$, $\lambda_{t}^{i}$, $W_t > 0$, and true inputs $l_{t}^{\star i}$, $h_{t}^{\star i}$, $c_{t}^{\star} > 0$ such that for all $s,t \in \mathcal{T}$ and all $i \in \{1,2\}$
    \begin{align*}
         U_{s}^{i} - U_{t}^{i} &\leq \lambda_{t}^{i} \Big[ w_{t}^{i} (l_{s}^{\star i} - l_{t}^{\star i}) + (q_{s}^{i} - q_{t}^{i}) + \mathcal{P}_{t}^{i}(Q_s - Q_t)
         + P_t^{i} ( W_s - W_t) \Big] \\
        \alpha_{1} &= \frac{w_{t}^1h_{t}^{\star 1}}{P_tW_t} \\
        \alpha_{2} &= \frac{w_{t}^2h_{t}^{\star 2}}{P_tW_t} \\
        \alpha_{3} &= \frac{w_{t}^1h_{t}^{\star 1}}{P_tW_t} \\
        \epsilon_t &= \log(W_t) - \alpha_{1}\log(h_{t}^{\star 1}) - \alpha_{2}\log(h_{t}^{\star 2}) - \alpha_{3}\log(c_{t}^{\star}),
    \end{align*}

    \noindent
    where the last equation is obtained from the natural logarithm of the production function equation. Suppose I have a solution 
    \begin{equation*}
        (U_{t}^{i}(r),\lambda_{t}^{i}(r), l_{t}^{\star i}(r), h_{t}^{\star i}(r), c_{t}^{\star}(r), P_t^{i}(r), \mathcal{P}_{t}^{i}(r))_{i \in \{1,2\}, t \in \mathcal{T}},
    \end{equation*}
    
    \noindent
    where $r$ denote the $r$th solution found by some solver. I provide a feasible algorithm that guarantees the next set of latent variables to be in the feasible space conditional on the data. The algorithm works provided the feasible space is nonempty in each block and standard regularity conditions associated with Gibbs samplers hold. 
    
    Intuitively, the idea is to recognize that it is difficult to uniformly sample all latent variables at once from the feasible space because the feasible space is complicated. Fortunately, I can break down the feasible space into conditional convex polytopes for which I can obtain closed-form bounds on the support of the latent variables. It thus becomes straightforward to uniformly sample a set of latent variables from each conditional convex polytope. \\




    \noindent
    \textbf{Step 1: Marginal Utility of Expenditure}

    \noindent
    Let $\Lambda = [1,L]$ denote the support of $\lambda_{t}^{i}$, where $L$ is an arbitrarily large number. Given a solution at step $r$, I want to find $\lambda_{t}^{i}(r+1)$ that satisfies
    \begin{align*}
        U_{s}^{i}(r) - U_{t}^{i}(r) &\leq \lambda_{t}^{i}(r+1)  \Big[ w_{t}^{i} (l_{s}^{\star i}(r) - l_{t}^{\star i}(r)) + (q_{s}^{i} -q_{t}^{i}) + \\
        &+ \mathcal{P}_{t}^{i}(r)(Q_s - Q_t) + P_t^{i}(r) \big( W_{s}(r) - W_{t}(r) \big) \Big].
    \end{align*}

    \noindent
    For convenience, let
    \begin{align*}
        denom^i := &\Big[ w_{t}^{i} (l_{s}^{\star i} - l_{t}^{\star i}) + (q_{s}^{i} - q_{t}^{i}) + \mathcal{P}_{t}^{i}(r)(Q_s - Q_t) + P_t^{i}(r) \big( W_{s}(r) - W_{t}(r) \big) \Big].
    \end{align*}

    \noindent
    It follows that
    \begin{align*}
        \lambda_{t}^{i}(r+1) \Delta \frac{U_{s}^{i}(r) - U_{t}^{i}(r)}{denom^i},
    \end{align*}
    
    \noindent
    where $\Delta := >$ if $denom^i > 0$ and $\Delta := <$ otherwise. Note that each $\lambda_{t}^{i}(r+1)$ has $T$ bounds. The greatest lower bound on $\lambda_{t}^{i}(r+1)$ is the maximum between one and the greatest lower bound. If there is no lower bound, then the greatest lower bound is one. Likewise, the least upper bound is the minimum between one and the least upper bound. If there is no upper bound, then the least upper bound is one. Draw $\lambda_{t}^{i}(r+1)$ uniformly over the support defined by the greatest lower bound and least upper bound. \\

    

    \noindent
    \textbf{Step 2: Children Welfare, Leisure, and Childcare}

    \noindent
    Conditional on the new solution $\lambda_{t}^{i}(r+1)$, I want children welfare, true leisure, and true childcare to be positive such that
    \begin{align}
        W_{t}(r) + \alpha \xi(W_{t}) &> 0 \\
        h_{t}^{\star 1}(r) + \alpha \xi(h_{t}^{\star 1}) &> 0 \\
        h_{t}^{\star 2}(r) + \alpha \xi(h_{t}^{\star 2}) &> 0 \\
        l_{t}^{\star 1}(r) + \alpha \xi(l_{t}^{\star 1}) &> 0 \\
        l_{t}^{\star 2}(r) + \alpha \xi(l_{t}^{\star 2}) &> 0.
    \end{align}
    
    \noindent
    These positivity constraints provide a set of inequality restrictions on $\alpha$. Next, I must further ensure that new leisure and new childcare of each household member satisfy the normalized time constraint
    \begin{align*}
        h_{t}^{\star i}(r) + \alpha \xi(h_{t}^{\star i}) + m_{t}^{i} + l_{t}^{\star i}(r) + \alpha \xi(l_{t}^{\star i}) = 1.
    \end{align*}

    \noindent
    This equation implies $\xi(h_{t}^{\star i}) = - \xi(l_{t}^{\star i})$, $i \in \{1,2\}$. That is, the direction taken for new childcare is the opposite of the direction for new leisure. 
    
    Next, it is important to ensure new children welfare and new childcare are consistent with a positive children expenditure:
    \begin{align*}
        w_{t}^{1}(h_{t}^{\star 1}(r) + \alpha \xi(h_{t}^{\star 1})) + w_{t}^{2}(h_{t}^{\star 2}(r) + \alpha \xi(h_{t}^{\star 2})) \leq P_{t}(r)(W_t(r) + \alpha \xi(W_t)).
    \end{align*}

    \noindent
    Rearranging, one obtains
    \begin{align*}
        \alpha \Delta \frac{ P_t(r)W_t(r) - w_{t}^{1}h_{t}^{\star 1}(r) - w_{t}^{2}h_{t}^{\star 2}(r) }{ w_{t}^{1} \xi(h_{t}^{\star 1}) +  w_{t}^{2} \xi(h_{t}^{\star 2}) - P_t \xi(W_t) },
    \end{align*}
    
    \noindent
    where $\Delta := <$ if $w_{t}^{1} \xi(h_{t}^{1}) +  w_{t}^{2} \xi(h_{t}^{2}) - P_t \xi(W_t) > 0$ and $\Delta := >$ otherwise. Further, I must also ensure that new children welfare and new childcare are compatible with decreasing returns to scale given children expenditure:
    \begin{align*}
        w_{t}^{1}(h_{t}^{\star 1}(r) + \alpha \xi(h_{t}^{\star 1})) + w_{t}^{2}(h_{t}^{\star 2}(r) + \alpha \xi(h_{t}^{\star 2})) + c_{t}(r) \leq P_{t}(r)(W_t(r) + \alpha \xi(W_t)).
    \end{align*}
    
    \noindent
    Rearranging, one obtains
    \begin{equation*}
        \alpha \Delta \frac{ P_t(r)W_t(r) - w_{t}^{1}h_{t}^{\star 1}(r) - w_{t}^{2}h_{t}^{\star 2}(r) - c_{t}(r) }{ w_{t}^{1} \xi(h_{t}^{\star 1}) +  w_{t}^{2} \xi(h_{t}^{\star 2}) - P_t \xi(W_t) },
    \end{equation*}

    \noindent
    where $\Delta := <$ if $w_{t}^{1} \xi(h_{t}^{1}) +  w_{t}^{2} \xi(h_{t}^{2}) - P_t \xi(W_t) > 0$ and $\Delta := >$ otherwise.
    
    Finally, I want new children welfare $W_{t}(r+1)$, new leisure $l_{t}^{\star i}(r+1)$, and new childcare $h_{t}^{\star i}(r+1)$ to satisfy 
    \begin{align*}
        U_{s}^{i}(r) - U_{t}^{i}(r) &\leq \lambda_{t}^{i}(r+1) \Big[ w_{t}^{i} (l_{s}^{\star i}(r) + \alpha \xi(l_{s}^{\star i}) - l_{t}^{\star i}(r) - \alpha \xi(l_{t}^{\star i})) + (q_{s}^{i} - q_{t}^{i}) + \\
        &+ \mathcal{P}_{t}^{i}(r)(Q_s - Q_t) + P_t^{i}(r) \big( W_s(r) + \alpha\xi(W_s) - W_t(r) - \alpha\xi(W_t) \big) \Big].
    \end{align*}

    \noindent
    This inequality can be rewritten as 
    \begin{align*}
        U_{s}^{i}(r) &- U_{t}^{i}(r) - \lambda_{t}^{i}(r+1) \Big[ w_{t}^{i} (l_{s}^{\star i}(r) - l_{t}^{\star i}(r)) + (q_{s}^{i} - q_{t}^{i}) + \mathcal{P}_{t}^{i}(r)(Q_s - Q_t) \\
        &+ P_t^{i}(r) \big( W_s(r) - W_{t}(r) \big) \Big] \\ &
        \leq \alpha \lambda_{t}^{i}(r+1) \big( P_t^{i}(r) (\xi(W_s) - \xi(W_t)) + w_{t}^{i}(\xi(l_{s}^{\star i}) - \xi(l_{t}^{\star i})) \big).
    \end{align*}

    \noindent
    Let $num^{i}$ denote the left-hand side of this inequality. Thus, I have
    \begin{align*}
        \alpha \Delta \frac{num^{i}}{ \lambda_{t}^{i}(r) \big( P_t^{i}(r)(\xi(W_s) - \xi(W_t)) + \xi(w_{t}^{i})(l_{s}^{\star i} - l_{t}^{\star i}) \big)},
    \end{align*}
    
    \noindent
    where $\Delta := >$ if $\lambda_{t}^{i}(r) \big( P_t^{i}(r) (\xi(W_s) - \xi(W_t)) + w_{t}^{i}(\xi(l_{s}^{\star i}) - \xi(l_{t}^{\star i})) \big) > 0$ and $\Delta := <$ otherwise. Draw $\alpha$ uniformly over its support as defined by the greatest lower bound and the least upper bound from the previous sets of inequalities. I obtain $(W_{t}(r+1),l_{t}^{\star i}(r+1),h_{t}^{\star i}(r+1))_{t \in \mathcal{T}}$ by picking $\alpha$ uniformly over its support defined by the greatest lower bound and least upper bound derived from the above inequalities. \\

    \noindent
    \textbf{Step 3: Utilities and Personalized Prices}

    \noindent
    I want personalized prices to be positive such that
    \begin{align}
        \mathcal{P}_{t}^{i}(r) + \beta \xi(\mathcal{P}_{t}^{i}) &> 0\\
        P_t^{i}(r) + \beta \xi(P_t^{i}) &> 0.
    \end{align}

    \noindent
    These inequalities can be transformed to get bounds on $\beta$:
    \begin{align}
         \beta  &\Delta -\frac{\mathcal{P}_t^{i}(r)}{\xi(\mathcal{P}_t^{i})} \label{support3} \\
         \beta  &\Delta -\frac{P_t^{i}(r)}{\xi(P_t^{i})} \label{support4},
    \end{align}

    \noindent
    where $\Delta := >$ if $\xi(\cdot) > 0$ and $\Delta := <$ otherwise. Next, similar to the previous step it is important to ensure new personalized prices for children welfare are consistent with a positive children expenditure:
    \begin{equation*}
        w_{t}^{1}h_{t}^{\star 1}(r+1) + w_{t}^{2}h_{t}^{\star 2}(r+1) \leq (P_t^{1}(r) + \beta \xi(P_t^1) + P_t^{2}(r) + \beta \xi(P_t^{2}))W_t(r+1).
    \end{equation*}

    \noindent
    Rearranging, one obtains
    \begin{equation*}
        \beta \Delta \frac{w_{t}^{1}h_{t}^{\star 1}(r+1) + w_{t}^{2}h_{t}^{\star 2}(r+1) - (P_{t}^{1}(r) + P_{t}^{2}(r))W_t(r+1)}{(\xi(P_{t}^{1}) + \xi(P_{t}^{2}))W_t(r+1)},
    \end{equation*}

    \noindent
    where $\Delta := >$ if $(\xi(P_{t}^{1}) + \xi(P_{t}^{2}))W_t(r+1) > 0$ and $\Delta := <$ otherwise. Further, I must also ensure that new personalized prices for children welfare are compatible with decreasing returns to scale given children expenditure:
        \begin{equation*}
        w_{t}^{1}h_{t}^{\star 1}(r+1) + w_{t}^{2}h_{t}^{\star 2}(r+1) + c_t(r) \leq (P_t^{1}(r) + \beta \xi(P_t^1) + P_t^{2}(r) + \beta \xi(P_t^{2}))W_t(r+1).
    \end{equation*}

    \noindent
    Rearranging, one obtains
    \begin{equation*}
        \beta \Delta \frac{w_{t}^{1}h_{t}^{\star 1}(r+1) + w_{t}^{2}h_{t}^{\star 2}(r+1) + c_t(r) - (P_{t}^{1}(r) + P_{t}^{2}(r))W_t(r+1)}{(\xi(P_{t}^{1}) + \xi(P_{t}^{2}))W_t(r+1)},
    \end{equation*}

    \noindent
    where $\Delta := >$ if $(\xi(P_{t}^{1}) + \xi(P_{t}^{2}))W_t(r+1) > 0$ and $\Delta := <$ otherwise.
    
    Finally, starting with the new numbers $\lambda_{t}^{i}(r+1)$, $W_{t}(r+1)$, $l_{t}^{\star i}(r+1)$, and $h_{t}^{\star i}(r+1)$, I want new utilities $U_{t}^{i}(r+1)$ and new personalized prices $\mathcal{P}_{t}^{i}(r+1)$, $P_{t}^{i}(r+1)$ to satisfy
    \begin{align*}
        U_{s}^{1}(r) &+ \beta \xi(U_{s}^{1}) - U_{t}^{1}(r) - \beta \xi(U_{t}^{1}) \leq \\ 
        &\lambda_{t}^{1}(r+1) \Big[ w_{t}^{1} (l_{s}^{\star 1}(r+1) - l_{t}^{\star 1}(r+1)) + (q_{s}^{1} - q_{t}^{1}) + \\
        &+ \big(\mathcal{P}_{t}^{1}(r) + \beta\xi(\mathcal{P}_{t}^{1})\big)\big(Q_s-Q_t\big) + \big( P_t^{1}(r) + \beta \xi(P_t^{1}) \big) \big( W_{s}(r+1) - W_{t}(r+1) \big) \Big],
    \end{align*}

    \noindent
    and
    \begin{align*}
        U_{s}^{2}(r) &+ \beta \xi(U_{s}^{2}) - U_{t}^{2}(r) - \beta \xi(U_{t}^{2}) \leq \\ 
        &\lambda_{t}^{2}(r+1) \Big[ w_{t}^{2} (l_{s}^{\star 2}(r+1) - l_{t}^{\star 2}(r+1)) + (q_{s}^{2} - q_{t}^{2}) + \\
        &+ \big(\mathcal{P}_{t}^{2}(r) + \beta\xi(\mathcal{P}_{t}^{2})\big)\big(Q_s-Q_t\big) + \big( P_t^{2}(r) + \beta \xi(P_t^{2}) \big) \big( W_{s}(r+1) - W_{t}(r+1) \big) \Big].
    \end{align*}

    \noindent
    With some algebra, I can rewrite the inequalities for household member $1$ as
    \begin{align*}
       \beta \Big[ \xi(U_{s}^{1}) &- \xi(U_{t}^{1}) + \lambda_{t}^{1}(r+1) \big[ 
       \xi(\mathcal{P}_{t}^{1})\big(Q_t - Q_s\big) + \xi(P_t^{1}) \big( W_{t}(r+1) - W_{s}(r+1) \big)  \big] \Big] \\ 
        &\leq U_{t}^{1}(r) - U_{s}^{1}(r) + \lambda_{t}^{1}(r+1) \Big[ w_{t}^{1} (l_{s}^{\star 1}(r+1) - l_{t}^{\star 1}(r+1)) +  (q_{s}^{1} - q_{t}^{1}) + \\
        &+ \mathcal{P}_{t}^{1}(r)\big(Q_s-Q_t\big) + P_t^{1}(r) \big( W_{s}(r+1) - W_{t}(r+1) \big) \Big],
    \end{align*}
    \noindent
    and the inequalities for household member $2$ as
    \begin{align*}
        \beta \Big[ \xi(U_{s}^{2}) &- \xi(U_{t}^{2}) + \lambda_{t}^{2}(r+1) \big[ \xi(\mathcal{P}_{t}^{2})\big(Q_t - Q_s\big) + \xi(P_t^{2}) \big( W_{t}(r+1) - W_{s}(r+1) \big) \big] \Big] \\
        &\leq U_{t}^{2}(r) - U_{s}^{2}(r) + \lambda_{t}^{2}(r+1) \Big[ w_{t}^{2} (l_{s}^{\star 2}(r+1) - l_{t}^{\star 2}(r+1)) +(q_{s}^{2} - q_t^{2}) + \\
        &+ \mathcal{P}_{t}^{2}(r)\big(Q_s-Q_t\big) + P_t^{2}(r) \big( W_{s}(r+1) - W_{t}(r+1) \big) \Big].
    \end{align*}

    \noindent
    For convenience, let
    \begin{align*}
     \phantom{denom^{1}}
        &\begin{aligned}
        \mathllap{num}^{1} &:= U_{t}^{1}(r) - U_{s}^{1}(r) + \lambda_{t}^{1}(r+1) \Big[ w_{t}^{1}(l_{s}^{\star 1}(r+1) - l_{t}^{\star 1}(r+1)) + (q_{s}^{1} - q_{t}^{1}) + \\
        &+ \mathcal{P}_{t}^{1}(r)\big(Q_s-Q_t\big) + P_t^{1}(r) \big( W_{s}(r+1) - W_{t}(r+1) \big) \Big]
        \end{aligned} \\
        &\begin{aligned}
        \mathllap{denom}^{1} &:= \xi(U_{s}^{1}) - \xi(U_{t}^{1}) + \\ &+ \lambda_{t}^{1}(r+1) \big[ \xi(\mathcal{P}_{t}^{1})\big(Q_t-Q_s\big) + \xi(P_t^{1}) \big( W_{t}(r+1) - W_{s}(r+1) \big)  \big]
        \end{aligned}
    \end{align*}

    \noindent
    and
    \begin{align*}
     \phantom{denom^{2}}
        &\begin{aligned}
        \mathllap{num}^{2} &:= U_{t}^{2}(r) - U_{s}^{2}(r) + \lambda_{t}^{2}(r+1) \Big[ w_{t}^{2} (l_{s}^{\star 2}(r+1) - l_{t}^{\star 2}(r+1)) + (q_s^{2} - q_t^{2}) + \\
        &+ \mathcal{P}_{t}^{2}(r)\big(Q_s - Q_t\big)+ P_t^{2}(r)  \big( W_{s}(r+1) - W_{t}(r+1) \big) \Big]
        \end{aligned} \\
        &\begin{aligned}
        \mathllap{denom}^{2} &:= \xi(U_{s}^{2}) - \xi(U_{t}^{2}) + \\ +& \lambda_{t}^{2}(r+1) \big[ \xi(\mathcal{P}_{t}^{2})\big(Q_t-Q_s\big) + \xi(P_t^{2}) \big( W_{t}(r+1) - W_{s}(r+1) \big) \big].
        \end{aligned}
    \end{align*}

    \noindent
    Therefore, I have
    \begin{align*}
        \beta \Delta \frac{num^{i}}{denom^{i}},
    \end{align*}

    \noindent
    where $\Delta := <$ if $denom^{i} > 0$ and $\Delta := >$ otherwise. I obtain $(U_{t}^{i}(r+1), \mathcal{P}_{t}^{i}(r+1), P_t^{i}(r+1))_{t \in \mathcal{T}}$ by picking $\beta$ uniformly over its support defined by the greatest lower bound and least upper bound derived from the above inequalities.\\

    \noindent
    \textbf{Step 4: Children Expenditure}

    \noindent
    I need to pick new true children expenditure that is positive such that
    \begin{align}
        c_{t}^{\star}(r) + \kappa \xi(c_{t}^{\star}) &> 0.
    \end{align}

    \noindent
    In addition, new children expenditure must yield decreasing returns to scale such that
    \begin{equation*}
        w_{t}^{1}h_{t}^{\star 1}(r+1) + w_{t}^{2}h_{t}^{\star 2}(r+1) + c_{t}^{\star}(r) + \kappa \xi(c_{t}^{\star}) \leq P_t(r+1)W_t(r+1).
    \end{equation*}

    \noindent
    Rearranging, one gets
    \begin{equation*}
        \kappa \Delta \frac{P_t(r+1)W_t(r+1) - w_{t}^{1}h_{t}^{\star 1}(r+1) -w_{t}^{2}h_{t}^{\star 2}(r+1) - c_{t}^{\star}(r)}{\xi(c_{t}^{\star})},
    \end{equation*}

    \noindent
    where $\Delta := <$ if $\xi(c_{t}^{\star}) > 0$ and $\Delta := >$ otherwise. I obtain $c_{t}^{\star}(r+1)$ by picking $\kappa$ uniformly over its defined by the greatest lower bound and least upper bound derived from the above inequalities. \\
    
    
    \noindent
    \textbf{Step 5: Production Parameters and Productivity}

    \noindent
    I am left with the task to recover output elasticities and productivity shocks. This requires no work as they are directly deduced from the first-order conditions and production function equation:
    \begin{align*}
        \alpha_{1,t}(r+1) &= \frac{w_{t}^{1}h_{t}^{\star 1}(r+1)}{P_t(r+1) W_{t}(r+1)} \\
        \alpha_{2,t}(r+1) &= \frac{w_{t}^{2}h_{t}^{\star 2}(r+1)}{P_t(r+1) W_{t}(r+1)} \\
        \alpha_{3,t}(r+1) &= \frac{c_{t}^{\star}(r+1)}{P_t(r+1) W_{t}(r+1)} \\
        \epsilon_t(r+1) &= \log(W_t(r+1)) - \alpha_{1,t}(r+1)\log(h_{t}^{\star 1}(r+1)) \\ &- \alpha_{2,t}(r+1)\log(h_{t}^{\star 2}(r+1)) - \alpha_{3,t}(r+1)\log(c_{t}^{\star}).
    \end{align*}

    \noindent
    This last step of the Gibbs sampler combined with the previous steps give a completely new solution to the model.\\

    \noindent
    \textbf{Sampling from the Feasible Space in $5$ Easy Steps}

    \noindent
    Suppose an initial solution $r = 0$ to the household problem is given (e.g., by solving a mixed-integer program). Then,
    \begin{enumerate}
        \item Given $r$, get $(\lambda_{t}^{i}(r+1))_{t \in \mathcal{T}}$ as outlined in Step $1$.
        \item Given $1$, get $(W_t(r+1), l_{t}^{\star i}(r+1), h_{t}^{\star i}(r+1))_{t \in \mathcal{T}}$ as outlined in Step $2$.
        \item Given $1$-$2$, get $(U_{t}^{i}(r+1), \mathcal{P}_{t}^{i}(r+1),P_t^{i}(r+1))_{t \in \mathcal{T}}$ as outlined in Step $3$.
        \item Given $1$-$3$, get $(c_{t}^{\star}(r+1))_{t \in \mathcal{T}}$ as outlined in Step $4$.
        \item Given $1$-$4$, get $(\bm{\alpha}_t(r+1), \epsilon_t(r+1))_{t \in \mathcal{T}}$ as outlined in Step $5$.
        \item Set $r = r+1$ and repeat $1$-$5$ until $r = R > 0$.\\
    \end{enumerate}

    \subsection{Miscellaneous}

    This subsection provides additional details about the sampling procedure.\\

    \noindent
    \textbf{Target Distribution}
    
    \noindent
    I ensure that the sampling procedure yields the desired least favorable distribution on measurement error by using a Metropolis-Hastings algorithm. Once a complete new solution is obtained from the Gibbs sampler, update the Markov chain with the appropriate acceptance ratio. Note that the acceptance ratio depends on the target distribution. In the application, the target distribution is proportional to a normal distribution:
    \begin{equation*}
        d\tilde{\eta}(\cdot|x_i) \propto \text{exp} \left(-||\bm{g}_{i}^{m}(x_i,e_i)||^{2} \right).
    \end{equation*}

    \noindent
    As pointed out by \cite{Schennach2014}, under mild regularity conditions the mean and variance of the distribution are inconsequential for the validity of Proposition \ref{propositionELVIS}. \\

    \noindent
    \textbf{Length of the Monte Carlo Markov Chain}

    \noindent
    The Gibbs sampler generates a Markov Chain that suffers from autocorrelation. For this reason, it is good practice to only keep a subset of the $R$ solutions, a technique known as thinning. In the application, I keep $5\%$ of all solutions. Also, the theory of stochastic processes tells us that convergence to the stationary distribution may take some time \textemdash its existence follows by construction of the Metropolis-Hastings algorithm. Accordingly, it is good practice to leave out the first few solutions. In the application, I leave out the first $100000$ solutions. I then draw another $100000$ solutions from the feasible space.

    \bibliographystyle{econ}
    \bibliography{references}

\end{document}